%% file: main-infinite.tex
\pdfoutput=1
\documentclass[1p,number,a4paper]{elsarticle}
\usepackage{amsfonts}
\usepackage{amsmath}
\usepackage{amssymb}
\usepackage{amsthm}
\usepackage[english]{babel}
\usepackage{complexity}
\usepackage{enumerate}
\usepackage{mathrsfs}
\usepackage{graphics}
\usepackage[pagewise]{lineno}
\usepackage{mathtools}
\usepackage{microtype}
\usepackage[defaultlines=4,all]{nowidow}
\usepackage[many]{tcolorbox}
\usepackage{todonotes}
\usepackage{xcolor}
\usepackage{xspace}

\usepackage[hidelinks]{hyperref}
\usepackage[nameinlink]{cleveref}

\usetikzlibrary{arrows, automata, fit, shapes, arrows.meta}

\tcolorboxenvironment{theorem}{
    colframe=cyan, sharp corners=west, 
    colback=cyan!10,
	rightrule=0pt,
	toprule=0pt,
	bottomrule=0pt,
	top=0pt,
	right=0pt,
	bottom=0pt,
 }
\tcolorboxenvironment{corollary}{
    colframe=blue, sharp corners=west, 
    colback=blue!10,
	rightrule=0pt,
	toprule=0pt,
	bottomrule=0pt,
	top=0pt,
	right=0pt,
	bottom=0pt,
 }
\tcolorboxenvironment{lemma}{
    colframe=red, sharp corners=west,
	colback=red!10,
	rightrule=0pt,
	toprule=0pt,
	bottomrule=0pt,
	top=0pt,
	right=0pt,
	bottom=0pt,
 }

\renewcommand{\epsilon}{\varepsilon}
\renewcommand{\phi}{\varphi}

\newcommand{\ie}{i.\,e.\xspace}
\newcommand{\eg}{e.\,g.\xspace}

\newcommand{\Amc}{\ensuremath{\mathcal{A}}\xspace}
\newcommand{\Bmc}{\ensuremath{\mathcal{B}}\xspace}
\newcommand{\Emc}{\ensuremath{\mathcal{E}}\xspace}
\newcommand{\Imc}{\ensuremath{\mathcal{I}}\xspace}
\newcommand{\Kmc}{\ensuremath{\mathcal{K}}\xspace}
\newcommand{\Lmc}{\ensuremath{\mathcal{L}}\xspace}
\newcommand{\Mmc}{\ensuremath{\mathcal{M}}\xspace}
\newcommand{\Pmc}{\ensuremath{\mathcal{P}}\xspace}

\newcommand{\Nbb}{\ensuremath{\mathbb{N}}\xspace}
\newcommand{\Zbb}{\ensuremath{\mathbb{Z}}\xspace}

\newcommand{\ebf}{\ensuremath{\mathbf{e}}\xspace}
\newcommand{\ubf}{\ensuremath{\mathbf{u}}\xspace}
\newcommand{\vbf}{\ensuremath{\mathbf{v}}\xspace}

\newcommand{\0}{\ensuremath{\mathbf{0}}\xspace}

\newcommand{\LPWA}{\ensuremath{\Lmc_\mathsf{PA}}\xspace}
\newcommand{\LPA}{\LPWA}
\newcommand{\LPPBA}{\ensuremath{\Lmc_\mathsf{Prefix}}\xspace}

\newcommand{\LReset}{\ensuremath{\Lmc_\mathsf{Reset}}\xspace}

\newcommand{\LPrefix}{\LPPBA}
\newcommand{\LPAomega}{\ensuremath{\Lmc_\mathsf{PA}^\omega}\xspace}

\theoremstyle{remark}
\newtheorem{theorem}{Theorem}[section]
\newtheorem{corollary}{Corollary}[section]
\newtheorem{definition}{Definition}[section]
\newtheorem{lemma}{Lemma}[section]

\newtheorem{remark}{Remark}[section]
\newtheorem{observation}{Observation}[section]

\newtheorem{example}{Example}
\crefname{corollary}{Corollary}{Corollaries}
\crefname{lemma}{Lemma}{Lemmas}
\crefname{section}{Section}{Sections}

\journal{arXiv}
\begin{document}
\begin{frontmatter}

\title{Parikh Automata on Infinite Words}
\author{Mario Grobler}
\ead{grobler@uni-bremen.de}
\address{University of Bremen, Bremen, Germany}
\author{Leif Sabellek} 
\ead{sabellek@uni-bremen.de}
\address{University of Bremen, Bremen, Germany}
\author{Sebastian Siebertz}\address{University of Bremen, Bremen, Germany}\ead{siebertz@uni-bremen.de}

\input{abstract}
\end{frontmatter}

\input{intro}

\input{prelims}

\input{definitions}

\input{epsilon}

\input{k-counter}

\input{expressiveness}

\input{decision}

\input{conclusion.tex}

\newpage
\bibliographystyle{plain}
\bibliography{lit}
\end{document}

%% file: abstract.tex
\begin{abstract}
Parikh automata on finite words were first introduced by Klaedtke and Rueß [Automata, Languages and Programming, 2003]. In this paper, we introduce several variants of Parikh automata on infinite words and study their expressiveness. We show that one of our new models is equivalent to synchronous blind counter machines introduced by Fernau and Stiebe [Fundamenta Informaticae, 2008]. All our models admit $\epsilon$-elimination, which to the best of our knowledge is an open question for blind counter automata. We then study the classical decision problems of the new automata models. 
\end{abstract}

%% file: intro.tex
\tableofcontents

\section{Introduction}

\noindent
Parikh automata on finite words (PA), originally introduced by Klaedtke and Rueß in~\cite{klaedtkeruess}, are finite automata enriched with counters. A PA is a non-deterministic finite automaton that is additionally equipped with a semi-linear set~$C$. Furthermore, every transition is equipped with a $d$-tuple of non-negative integers and every time a transition is used, the counters are incremented by the values in the tuple accordingly. An input word is accepted if the PA ends in a final state and additionally, the resulting $d$-tuple in the counters lies in~$C$. 
The class of languages recognized by PA contains all regular languages, but also many more, even languages that are not context-free, \eg, the language $\{a^nb^nc^n \mid n \in \mathbb{N}\}$. On the other hand, the language of palindromes is context-free, but cannot be recognized by a PA.

On finite words, Parikh automata have been investigated extensively. As shown in~\citep{klaedtkeruess}, the class of languages recognized by PA is captured precisely by weak existential monadic second-order logic of one successor extended with linear cardinality constraints. Cadilhac, Finkel and McKenzie introduced the variant of affine Parikh automata \cite{DBLP:journals/ita/CadilhacFM12} and used Parikh automata to characterize bounded languages \cite{DBLP:journals/ijfcs/CadilhacFM12}. Two-way PA were studied in~\cite{DBLP:conf/fossacs/DartoisFT19} and PA with a visibly pushdown stack were studied in~\cite{FiliotGM19}. 

In this paper, we initiate the research on Parikh automata on infinite words. This research direction was proposed by Klaedtke and Rueß in the conclusion of their work~\cite{klaedtkeruess}. Automata on infinite words play an important role in logics and formal verification, see e.g.\ the textbook~\cite{thomas2002automata}. 
For example, a non-terminating system that should distribute resources among three consumers equally could be modelled using the language $\{a^nb^nc^n \mid n \in \mathbb{N}\}^\omega$, which is not $\omega$-regular and not even $\omega$-context-free. By Büchi's theorem every regular $\omega$-language can be characterized as the finite union of $U_iV_i^\omega$, where $U_i,V_i$ are regular languages. 
A natural generalization is the class $\LPAomega$ of languages that can be characterized as the union of $U_iV_i^\omega$, where $U_i,V_i$ are Parikh-recognizable languages. This class was also mentioned by Fernau and Stiebe under the name $\Kmc_*$ in~\cite{blindcounter}. 

We suggest three different definitions for extending Büchi automata \cite{buechi} with a Parikh condition, which we call \emph{Parikh-Büchi-automata} (PBA).
All models are syntactically equal to PA, and they only differ in their semantics. 
We consider \emph{PBA with prefix-acceptance condition} (PPBA), \emph{PBA with strong reset-acceptance condition} (SPBA), and \emph{PBA with weak reset-acceptance condition} (WPBA). Furthermore, we consider their variants with $\epsilon$-transitions and variants where every accepting state is equipped with its own semi-linear~set. 

A PPBA accepts its input if there is a run that satisfies the acceptance condition for infinitely many prefixes of the word, \ie, it happens infinitely often that the automaton is in an accepting state and the current sum of the vectors lies in the semi-linear set of the automaton.
This model was proposed by Klaedtke and Rueß~\cite{klaedtkeruess}. We write $\LPrefix$ for the class of PPBA-recognizable languages. 
We prove that PPBA are equivalent to synchronous blind counter machines introduced by Fernau and Stiebe~\cite{blindcounter}. 
The main difficulty in this result is to eliminate $\epsilon$-transitions in PPBAs, which turns out to be surprisingly challenging. 
By this equivalence, PPBA-recognizable languages are closed under union, but not under intersection or complement. 
Furthermore, if $L$ is a Parikh language, then~$L^\omega$ is not necessarily PPBA-recognizable, \eg, the language $\{a^nb^n \mid n \in \mathbb{N}\}^\omega$ is not PPBA-recognizable. 

This last weakness motivates the definition of PBA-acceptance with the ability to reset the values of all counters, which leads to the definitions of SPBA and WPBA. 
An SPBA is defined like a PPBA, but the current sum of the vectors \emph{must} lie in the semi-linear set whenever an accepting state is visited. 
After the visit of an accepting state all counters are reset to zero. Similarly, a WPBA is not forced to, but \emph{may} reset in an accepting state. Both automata accept if they reset infinitely often. We prove that both models have the same expressiveness and we
write $\LReset$ for the class of SPBA/WPBA-recongizable $\omega$-languages. 
One of our main results is that PBA with a reset condition are strictly more expressive than PPBA by showing that $\LPAomega \subsetneq \LReset$. Together with the known results for blind counter automata we get 
\[\LPrefix \subsetneq \LPAomega \subsetneq \LReset.\]

If $L$ is a Parikh language, then $L^\omega$ is SPBA/WPBA-recognizable, hence, they do not have the weakness of PPBA mentioned above.

Finally, we study the classical decision problems for PPBA and SPBA/WPBA. We show that for all models, emptiness is $\coNP$-complete and universality (and hence, inclusion and equivalence) are undecidable. 

\medskip
\noindent
\textbf{Organization. }
After giving the necessary background in \Cref{sec:prelims}, we give our definitions of Parikh automata on infinite words and make some simple observations in \Cref{sec:definition}. In \Cref{sec:epsilon-elimination} we show how to elimination $\epsilon$-transitions. In \Cref{sec:blind-counter} we prove the equivalence of PPBA and blind counter automata and discuss some implications. In \Cref{sec:characterization}, we prove $\LPAomega \subsetneq \LReset$. In \Cref{sec:decision}, we study the decision problems for PPBA and SPBA/WPBA. We conclude in \Cref{sec:conclusion}.




%% file: prelims.tex
\section{Preliminaries}
\label{sec:prelims}

\subsection{Finite and infinite words}
We write $\Nbb$ for the set of non-negative integers including $0$, and $\Zbb$ for the set of all integers. Let $\Sigma$ be an alphabet, \ie, a finite non-empty set and let $\Sigma^*$ be the set of all finite words over $\Sigma$. 
For a word $w \in \Sigma^*$, we denote by $|w|$ the length of $w$, and by $|w|_a$ the number of occurrences of the letter $a \in \Sigma$ in $w$. 
We write $\varepsilon$ for the empty word of length~$0$ and $\Sigma^+ = \Sigma^* \setminus \{\varepsilon\}$ for the set of all non-empty finite words over $\Sigma$.

An \emph{infinite word} over an alphabet $\Sigma$ is a function $\alpha : \Nbb \setminus \{0\} \rightarrow \Sigma$. We often write~$\alpha_i$ instead of~$\alpha(i)$. 
Thus, we can understand an infinite word as an infinite sequence of symbols $\alpha = \alpha_1\alpha_2\alpha_3\ldots$ For $m \leq n$, we abbreviate the finite infix $\alpha_m \ldots \alpha_n$ by $\alpha[m,n]$. 
We denote by $\Sigma^\omega$ the set of all infinite words over $\Sigma$. 
We call a subset $L \subseteq \Sigma^\omega$ an \emph{$\omega$-language}. 
Moreover, for $L \subseteq \Sigma^*$, we define $L^\omega = \{w_1w_2\dots \mid w_i \in L \setminus \{\varepsilon\}\} \subseteq \Sigma^\omega$.

\subsection{Regular and $\omega$-regular languages}

A \emph{Nondeterministic Finite Automaton} (NFA) is a tuple $\Amc = (Q, \Sigma, q_0, \Delta, F)$, where~$Q$ is the finite set of states, $\Sigma$ is the input alphabet, $q_0 \in Q$ is the initial state, $\Delta \subseteq Q \times \Sigma \times Q$ is the set of transitions and $F \subseteq Q$ is the set of accepting states. 
A \emph{run} of $\Amc$ on a word $w = w_1 \cdots w_n\in \Sigma^*$ is a (possibly empty) sequence of transitions $r = r_1 \dots r_n$ with $r_i = (p_{i-1}, w_i, p_i)\in \Delta$ such that $p_0=q_0$. 
We say $r$ is \emph{accepting} if $p_n \in F$. The empty run on~$\epsilon$ is accepting if $q_0 \in F$. We define the \emph{language recognized by \Amc} as $L(\Amc) = \{w \in \Sigma^* \mid \text{there is an accepting run of $\Amc$ on $w$}\}$. If a language $L$ is recognized by some NFA $\Amc$, we call $L$ \emph{regular}.

A \emph{Büchi Automaton} (BA) is an NFA $\Amc = (Q, \Sigma, q_0, \Delta, F)$ that takes infinite words as input. 
A \emph{run} of $\Amc$ on an infinite word $\alpha_1\alpha_2\alpha_3\dots$ is an infinite sequence of transitions $r = r_1 r_2 r_3 \dots$ with $r_i = (p_{i-1}, \alpha_i, p_i) \in \Delta$ such that $p_0=q_0$. 
We say $r$ is \emph{accepting} if there are infinitely many~$i$ such that $p_i \in F$. 
We define the \emph{$\omega$-language recognized by~$\Amc$} as $L_\omega(\Amc) = \{\alpha \in \Sigma^\omega \mid \text{there is an accepting run of $\Amc$ on $\alpha$}\}$. 
If an $\omega$-language $L$ is recognized by some BA $\Amc$, we call $L$ \emph{$\omega$-regular}. Büchi's theorem establishes an important connection between regular and $\omega$-regular languages:
\begin{theorem}[Büchi]
\label{thm:buechi}
A language $L \subseteq \Sigma^\omega$ is $\omega$-regular if and only if there are regular languages $U_1, V_1, \dots, U_n, V_n \subseteq \Sigma^*$ for some $n \geq 1$ such that $L = U_1V_1^\omega \cup \dots \cup U_nV_n^\omega$.
\end{theorem}

%

\subsection{Semi-linear sets}
A \emph{linear set} of dimension $d$ for $d \geq 1$ is a set of the form $\{b_0 + b_1z_1 + \dots + b_\ell z_\ell \mid z_1, \dots, z_\ell \in \Nbb\} \subseteq \Nbb^d$ for $b_0,\ldots, b_\ell\in \Nbb^d$.
A \emph{semi-linear set} is the finite union of linear sets.
For vectors $\ubf = (u_1, \dots, u_c)\in \Nbb^c, \vbf = (v_1, \dots, v_d) \in \Nbb^d$, we denote by $\ubf \cdot \vbf = (u_1, \dots, u_c, v_1, \dots, v_d) \in \Nbb^{c+d}$ the \emph{concatenation of $\ubf$ and $\vbf$}. 
We extend this definition to sets of vectors. Let $U \subseteq \Nbb^c$ and $V \subseteq \Nbb^d$. Then $U \cdot V = \{\ubf \cdot \vbf \mid \ubf \in U, \vbf \in V\} \subseteq \Nbb^{c+d}$. 
We denote by $\0^d$ (or simply $\0$ if $d$ is clear from the context) the all-zero vector, and by $\ebf^d_i$ (or simply $\ebf_i)$ the $d$-dimensional vector where the $i$th entry is~$1$ and all other entries are 0.

\subsection{Parikh-recognizable languages}

A \emph{Parikh Automaton} (PA) is a tuple $\Amc = (Q, \Sigma, q_0, \Delta, F, C)$ where $Q$, $\Sigma$, $q_0$, and~$F$ are defined as for NFA, $\Delta \subseteq Q \times \Sigma \times \Nbb^d \times Q$ is the set of \emph{labeled transitions}, and $C \subseteq \Nbb^d$ is a semi-linear set. 
We call $d$ the \emph{dimension} of $\Amc$ and refer to the entries of a vector $\vbf$ in a transition $(p, a, \vbf, q)$ as \emph{counters}.
Similar to NFA, a \emph{run} of $\Amc$ on a word $w = x_1 \dots x_n$ is a (possibly empty) sequence of labeled transitions $r = r_1 \dots r_n$ with $r_i = (p_{i-1}, x_i, \vbf_i, p_i) \in \Delta$ such that $p_0 = q_0$. We define the \emph{extended Parikh image} of a run $r$ as $\rho(r) = \sum_{i \leq n} \vbf_i$ (with the convention that the empty sum equals $\0$).
We say $r$ is accepting if $p_n \in F$ and $\rho(r) \in C$,
referring to the latter condition as the \emph{Parikh condition}. 
We define the \emph{language recognized by \Amc} as $L(\Amc) = \{w \in \Sigma^* \mid \text{there is an accepting run of $\Amc$ on $w$}\}$. 
If a language $L\subseteq \Sigma^*$ is recognized by some PA, then we call $L$ \emph{Parikh-recognizable}. We write $\LPA$ to denote the class of all Parikh-recognizable languages.

%% file: definitions.tex
\section{Parikh automata on infinite words}
\label{sec:definition}

In this section we introduce our models of Parikh automata on infinite words and make some simple observations. 

\subsection{Definitions of Parikh-Büchi Automata}

We begin with the definition of Parikh automata on infinite words.

\begin{definition}
\label{def:pba}
A \emph{Parikh-Büchi Automaton} (PBA) is a PA $\Amc=(Q, \Sigma, q_0, \Delta, F, C)$. A run of $\Amc$ on an infinite word $\alpha_1 \alpha_2 \alpha_3 \dots$ is an infinite sequence of labeled transitions $r = r_1 r_2 r_3 \dots$ with $r_i = (p_{i-1}, \alpha_i, \vbf_i, p_i) \in \Delta$ such that $p_0 = q_0$. 
We consider three acceptance conditions.
\begin{enumerate}
\item \emph{PBA with prefix-acceptance condition} (PPBA): We say that $r$ \emph{satisfies the prefix-acceptance condition} (\emph{is accepting} for short) if there are infinitely many $i \geq 1$ such that $p_i \in F$ and $\rho(r_1 \dots r_i) \in C$. For these positions $i$ we say that there is an \emph{accepting hit in $r_i$}. We define the $\omega$-language recognized by a PPBA $\Amc$ as $P_\omega(\Amc) = \{\alpha \in \Sigma^\omega \mid  \Amc \text{ accepts } \alpha\}$.

\item \emph{PBA with weak reset-acceptance condition} (WPBA): We say that $r$ \emph{satisfies the weak reset-acceptance condition} (\emph{is accepting} for short) if there are infinitely many \emph{reset positions} $0 = k_0 < k_1 < k_2\dots$ such that $p_{k_i} \in F$ and $\rho(r_{k_{i-1} + 1} \dots r_{k_i}) \in C$ for all $i \geq 1$. We define the $\omega$-language recognized by a WPBA $\Amc$ as $W_\omega(\Amc) = \{\alpha \in \Sigma^\omega \mid \Amc \text{ accepts } \alpha\}$.

\item \emph{PBA with strong reset-acceptance condition} (SPBA): Let $k_0 = 0$ and denote by $k_1, k_2, \dots$ the positions of all accepting states in $r$, \ie, $p_{k_i} \in F$ for all $i \geq 1$. We say that $r$ \emph{satisfies the strong reset-acceptance condition} (\emph{is accepting} for short) if $k_1, k_2, \dots$ is an infinite sequence and $\rho(r_{k_{i-1} + 1} \dots r_{k_i}) \in C$ for all $i \geq 1$. We define the $\omega$-language recognized by an SPBA $\Amc$ as $S_\omega(\Amc) = \{\alpha \in \Sigma^\omega \mid \Amc \text{ accepts } \alpha\}$.
\end{enumerate}
\end{definition}

Intuitively, an SPBA $\Amc$ accepts all infinite words $\alpha$ in such a way, that whenever a run of $\Amc$ on $\alpha$ visits an accepting state, the Parikh condition \emph{must} be satisfied. After that, the counters are reset. 
Compared to that, a WPBA \emph{may} reset the counters when visiting an accepting state but does not have to do so.
Both accept if they reset their counters infinitely often. A PPBA does not reset its counters at all. It accepts if there exists a run that has infinitely many accepting (finite) prefixes that satisfy the Parikh condition. 
Since all automata are syntactically equal objects we write $P_\omega$, $W_\omega$ and $S_\omega$ for the accepted $\omega$-languages to make the different acceptance conditions explicit. For the same reason we will often speak about resetting states instead of accepting states in the context of SPBA/WPBA. 

\begin{example}
\label{ex:pba}
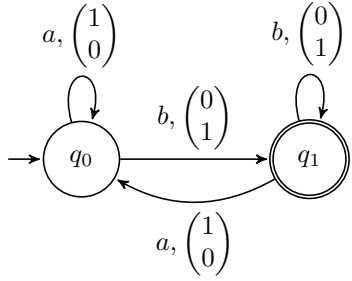
\begin{figure}
	\centering
	\begin{tikzpicture}[->,>=stealth',shorten >=1pt,auto,node distance=3cm, semithick]
	\tikzstyle{every state}=[minimum size=1cm]
	
	\node[initial, initial text={}, state] (q0) {$q_0$};
	\node[state, accepting] (q1) [right of=q0] {$q_1$};	
	
	\path
	(q0) edge [loop above] node {$a, \begin{pmatrix}1\\0\end{pmatrix}$} (q0)
	(q0) edge              node {$b, \begin{pmatrix}0\\1\end{pmatrix}$} (q1)
	(q1) edge [loop above] node {$b, \begin{pmatrix}0\\1\end{pmatrix}$} (q1)
	(q1) edge [bend left]  node {$a, \begin{pmatrix}1\\0\end{pmatrix}$} (q0)	
	;
	\end{tikzpicture}
	\caption{The PBA $\Amc$ with $C=\{(z,z) \mid z \in \Nbb\})$ from \Cref{ex:pba}.}
	\label{fig:pba}
\end{figure}

Let $\Amc$ with $C = \{(z,z) \mid z \in \Nbb\}$ be the PBA depicted in Figure \ref{fig:pba}. Then $W_\omega(\Amc) = \{w \in \{a,b\}^* \cdot \{b\} \mid |w|_a = |w|_b\}^\omega$. 
Note that this language is not $\omega$-regular. While the first counter tracks the number of already read $a$s, the second counter tracks the number of already read $b$s. 
The semi-linear set $C$ essentially states that both values need to be the same (eventually). 
Since a WPBA may or may not reset its counters when visiting an accepting state, $\Amc$ can "move freely" between both of its states. Note that $P_\omega(\Amc)=W_\omega(\Amc)$.

Considering the strong reset-acceptance condition, we have $S_\omega(\Amc) = \{ab\}^\omega$, only a single infinite word is accepted by this automaton. Assume $\Amc$ tries to accept an infinite word $\alpha$ that has an infix $a^n$ where $n > 1$. 
That is, $\Amc$ loops $n$ times in~$q_0$, hence the first counter has a value of $n$ and the second counter a value of 0. 
The only chance to increase the second counter is by moving into the accepting state $q_1$. 
Since both counters do not have the same values by then, $\alpha$ will be rejected. 
If $\Amc$ tries to accept an infinite word that starts with a $b$, or has an infix $b^n$ where $n > 1$, the input will be rejected for similar reasons.
\end{example}

We now introduce PBAs that allow different semi-linear sets on all accepting states, which we call Multi-PBA. 

\begin{definition}
A \emph{Multi-PBA} is a tuple $\Amc = (Q, \Sigma, q_0, \Delta, F, C)$ where $C : F \rightarrow \Pmc(\Nbb^d)$ is a function that assigns a semi-linear set to each $q \in F$.
We define Multi-PPBA (MPBA), MWPBA and MSPBA as SPBA, WPBA and SPBA, respectively, where the condition $\rho(r_i,\ldots, r_j)\in C$ in the respective definitions is replaced by $\rho(r_i,\ldots, r_j)\in C(p_j)$. 
\end{definition}

We finally introduce definitions of PBA that allow $\varepsilon$-transitions. 

\begin{definition}
An $\epsilon$-PBA is a tuple $\Amc = (Q, \Sigma, q_0, \Delta, \Emc, F, C)$ where $\Emc \subseteq Q \times \{\varepsilon\} \times \Nbb^d \times Q$ is the set of \emph{labeled $\varepsilon$-transitions}, and all other entries are defined as for PBA. 
A run of $\Amc$ on an infinite word $\alpha_1\alpha_2\alpha_3 \dots$ is an infinite sequence of transitions $r \in (\Emc^* \Delta)^\omega$, say $r = r_1r_2r_3 \dots$ with $r_i = (p_{i-1}, \gamma_i, \vbf_i, p_i)$ such that $p_0 = q_0$, and $\gamma_i = \varepsilon$ if $r_i \in \Emc$, and $\gamma_i = \alpha_j$ if $r_i \in \Delta$ is the $j$-th occurrence of a (non-$\varepsilon$) transition in $r$. 
The definitions of prefix-acceptance condition, weak reset-acceptance condition and strong reset-condition are the same as in Definition \ref{def:pba} and extend to Multi-PBA in the natural way.
We use the terms $\varepsilon$-PPBA, $\varepsilon$-WPBA, $\varepsilon$-SPBA, $\varepsilon$-MPPBA, $\varepsilon$-MWPBA, and $\varepsilon$-MSPBA, respectively. We call a transition of the form $(p, \epsilon, \vbf, p) \in \Emc$ an \emph{$\epsilon$-loop}.
\end{definition}


Note that we can treat every PBA as an $\epsilon$-PBA, that is, a PBA $\Amc = (Q,\Sigma, q_0, \Delta, F, C)$ is equivalent to the $\varepsilon$-PBA $\Amc' = (Q, \Sigma, q_0, \Delta, \varnothing, F, C)$ in the sense that $P_\omega(\Amc) = P_\omega(\Amc')$, $W_\omega(\Amc) = W_\omega(\Amc')$, and $S_\omega(\Amc) = S_\omega(\Amc')$.

\subsection{Simple Observations.}
We begin with a few simple observations. We first show that the accepting states $F$ and the linear sets of $C$ of an $\varepsilon$-PPBA $\Amc$ are in a sense independent, as formalized by the following lemma.
\begin{lemma}
\label{lemma:indep}
Let $\Amc = (Q, \Sigma, q_0, \Delta, \Emc, F, C)$ be an $\varepsilon$-PPBA, where $C = C_1 \cup \dots \cup C_\ell$ for linear sets~$C_j$, $1\leq j \leq \ell$. Then $P_\omega(\Amc) = \bigcup_{f \in F} \bigcup_{1\leq j \leq \ell} P_\omega(Q, \Sigma, q_0, \Delta, \Emc, \{f\}, C_j)$.
\end{lemma}
\begin{proof}
The right-to-left direction is obvious, thus we show the left-to-right direction.
Let $\alpha \in P_\omega(\Amc)$ with an accepting run $r = r_1 r_2 r_3 \dots$ where $r_i = (p_{i-1}, \gamma_i, \vbf_i, p_i)$, \ie, there are infinitely many $i \geq 1$ such that $p_i \in F$ and $\rho(r_1 \dots r_i) \in C$. 
By the infinite pigeonhole principle, there is a state $f \in F$ such that $p_{i'} = f$ and $j$ such that $\rho(r_1 \dots r_{i'}) \in C_j$ for infinitely many $i' \geq 1$. 
Thus, $\alpha\in P_\omega(Q, \Sigma, q_0, \Delta, \Emc, \{f\}, C_j)$. 
\end{proof}

Note that this lemma does in general not hold for $\varepsilon$-WPBA and $\varepsilon$-SPBA.


The following lemma illustrates the simple and yet important combinatorial method to use additional counters and an adapted semi-linear set to store information about runs. 

\begin{lemma}
\label{lemma:GPPBA}
For every $\varepsilon$-MPPBA $\Amc$ there exists an $\varepsilon$-PPBA $\Amc'$ with a single accepting state such that $P_\omega(\Amc) = P_\omega(\Amc')$. If $\Amc$ has no non-loop $\epsilon$-transitions, then $\Amc'$ has no non-loop $\epsilon$-transitions and if $\Amc$ is an MPPBA, then $\Amc'$ is a PPBA. 
\end{lemma}

\begin{proof}

Let $\Amc = (Q, \Sigma, q_0, \Delta, \Emc, F, C)$ be an $\varepsilon$-MPPBA of dimension $d$ and let $F = \{q_1, \dots, q_k\}$. Since $\Amc$ is an $\epsilon$-MPPBA, $C\colon F\rightarrow \mathcal{P}(\Nbb^d)$ is a function assigning a semi-linear set to each $q\in F$. 
We construct an $\varepsilon$-PPBA $\Amc'$ of dimension $d+k$ with a single accepting state. 
We introduce a new accepting state $q_f$. 
A naive approach would be to connect all~$q_i$ via $\varepsilon$-transitions to $q_f$ and vice versa. However, this approach fails, as $q_f$ has no information from which of the accepting states a run enters and thus does not know which outgoing transitions are valid such that no invalid shortcuts are created. Hence, we use non-determinism to guess one accepting state (say $q_i$) that is visited infinitely often to satisfy the prefix-acceptance condition. Additionally, we introduce one new counter per accepting state to ensure that we only use transitions related to $q_i$.

Formally, let $\Amc' = (Q \cup \{q_f\}, \Sigma, \Delta', \Emc', \{q_f\}, C)$ where
\begin{align*}
    \Delta' =&\ \{(q,a,\vbf \cdot \0^k, q') \mid (q,a,\vbf, q') \in \Delta\} \\
            \cup&\ \{(q,a,\vbf \cdot \ebf_i^k, q_f) \mid (q,a,\vbf,q_i) \in \Delta, i \leq k\} \\
            \cup&\ \{(q_f,a,\vbf \cdot \ebf_i^k, q) \mid (q_i,a,\vbf,q) \in \Delta, i \leq k\} \\
            \cup&\ \{(q_f, a, \vbf \cdot \ebf_i^k, q_f) \mid (q_i, a, \vbf, q_i) \in \Delta, i \leq k\}
\end{align*}
and
\begin{align*}
    \Emc' =&\ \{(q,\varepsilon,\vbf \cdot \0^k, q') \mid (q,\varepsilon,\vbf, q') \in \Emc\} \\
            \cup&\ \{(q,\varepsilon,\vbf \cdot \ebf_i^k, q_f) \mid (q,\varepsilon,\vbf,q_i) \in \Emc, i \leq k, q \neq q_i\} \\
            \cup&\ \{(q_f,\varepsilon,\vbf \cdot \ebf_i^k, q) \mid (q_i,\varepsilon,\vbf,q) \in \Emc, i \leq k, q \neq q_i\} \\
            \cup&\ \{(q_f, \varepsilon, \vbf \cdot \ebf_i^k, q_f) \mid (q_i, \varepsilon, \vbf, q_i) \in \Emc, i \leq k\}.
\end{align*}

Define $C_i = C(q_i) \cdot \{z \ebf_i^k \mid z \in \Nbb\}$ and $C = \bigcup_{i \leq k} C_i$. We claim that $P_\omega(\Amc) = P_\omega(\Amc')$.

To show $P_\omega(\Amc) \subseteq P_\omega(\Amc')$, let $\alpha \in P_\omega(\Amc)$ with accepting run $r = r_1 r_2 r_3 \dots$ of $\Amc$, where $r_j = (p_{j-1}, \gamma_j, \vbf_j, p_j)$. 
By the infinite pigeonhole principle, there exists an accepting state $q_i$ satisfying the prefix-acceptance condition infinitely often, \ie, there are infinitely many $j$ such that $p_j = q_i$ and $\rho(r_1 \dots r_j) \in C(q_i)$. 
From $r$ we construct an accepting run $r' = r'_1 r'_2 r'_3 \dots$ of $\Amc'$ by replacing every transition after reading $\alpha_1$ of the form $(p_{j-1}, \gamma_j, \vbf_j, q_i)$ with $(p_{j-1}, \gamma_j, \vbf_j \cdot \ebf^k_i, q_f)$, every $(q_i, \gamma_j, \vbf_j, q_i)$ with $(q_f, \gamma_j, \vbf_j \cdot \ebf^k_i, q_f)$, and every $(q_i, \gamma_j, \vbf_j, p_j)$ with $(q_f, \gamma_j, \vbf_j \cdot \ebf^k_i, p_j)$. 
Finally, every other transition $(p_{j-1}, \gamma_j, \vbf_j, p_j)$ is replaced by $(p_{j-1}, \gamma_j, \vbf_j \cdot \0^k, p_j)$. 
By the choice of $\Delta'$ and $\Emc'$ the run $r'$ is indeed a run of $\Amc'$ on $\alpha$. 
Furthermore, $r'$ is accepting, as every accepting hit in $q_i$ in $r$ after reading~$\alpha_1$ (before reading the first symbol of $\alpha$ it might be the case that we could not shortcut into~$q_f$) translates into an accepting hit in $r'$: the first $d$ counter values in both runs are equal at every position, and by construction, there is only a single non-zero value in the $i$-th component of the appended $k$-dimensional vector in $r'$. 
Thus, we have $\rho(r'_1 \dots r'_j) \in C_i \subseteq C$ whenever we have $\rho(r_1 \dots r_j) \in C(q_i)$.

To show $P_\omega(\Amc) \supseteq P_\omega(\Amc')$, let $\alpha \in P_\omega(\Amc')$ with accepting run $r' =  r'_1 r'_2 r'_3 \dots$ of~$\Amc'$ where $r'_j = (p_{j-1}, \gamma_j, \vbf_j, p_j)$. 
By construction of $C$ there is a unique $C_i$ such that $\rho(r'_1 \dots r'_j) \in C_i$ for infinitely many $j$. 
Similar as above, we can replace every occurrence of $q_f$ in $r'$ by $q_i$ and forget the additional counters to obtain a valid run $r$ of $\Amc$ on $\alpha$, which is accepting as every accepting hit on $q_f$ in $r'$ then corresponds to an accepting hit on $q_i$ in $r$.

Finally, observe that $\Amc'$ has only non-loop $\varepsilon$-transitions if $\Amc$ has only non-loop \mbox{$\varepsilon$-transitions}, as in this case the first and fourth line of the definition of $\Emc'$ introduce only $\varepsilon$-loops and the sets defined in the second and third line are empty. If $\Amc$ was a MPPBA, we may consider it as an $\varepsilon$-MPPBA with $\Emc' = \varnothing$, hence, $\Amc'$ can be seen as a PPBA in this case.
\end{proof}

Observe that we can use a similar trick to convert ($\varepsilon$-)MSPBA and ($\varepsilon$-)MWPBA into equivalent ($\varepsilon$-)SPBA and ($\varepsilon$-)WPBA. However, in general a single accepting state is not sufficient anymore.

\begin{lemma}
\label{lem:MSPBAtoSPBA}
For every $\varepsilon$-MSPBA ($\varepsilon$-MWPBA)~$\Amc$ there exists an $\varepsilon$-SPBA ($\varepsilon$-WPBA)~$\Amc'$ such that $S_\omega(\Amc) = S_\omega(\Amc')$ ($W_\omega(\Amc) = W_\omega(\Amc')$). If $\Amc$ has no non-loop $\epsilon$-transitions, then~$\Amc'$ has no non-loop $\epsilon$-transitions and if $\Amc$ is an MSPBA (MWPBA), then $\Amc'$ is an SPBA (WPBA). 
\end{lemma}

\begin{proof}[Proof sketch.]
    Let $\Amc = (Q, \Sigma, q_0, \Delta, \Emc, F, C)$ be an ($\varepsilon$-)MSPBA of dimension $d$ and let $F = \{q_1, \dots, q_k\}$. Again, $C\colon F\rightarrow \mathcal{P}(\Nbb^d)$ is a function assigning a semi-linear set to each $q\in F$. 
We construct an equivalent $(\varepsilon$-)SPBA $\Amc'$ of dimension $d+k$ with a single semi-linear set $C'\subseteq \Nbb^{d+k}$. 

The idea is to use the additional counters to mark into which resetting state we enter, that is, every transition that does not lead to a resetting state is simply padded with~$0$, while every transition entering the resetting state $q_i$ has a $1$ exactly at the $d+i$th position. The semi-linear set $C'$ is defined as $\bigcup_{1\leq i\leq k}C(q_i)\cdot \{\ebf^d_i\}$. 

When $\Amc$ is an $(\epsilon)$-MWPBA, then additionally, for every transition entering a resetting state we add the transition padded with $0$. Then, the resulting automaton can non-deterministically decide if it wants to reset or not. 

Finally, observe that the resulting automaton $\Amc'$ has only non-loop $\varepsilon$-transitions if $\Amc$ has only non-loop \mbox{$\varepsilon$-transitions}. If $\Amc$ was a MPPBA, we may consider it as an $\varepsilon$-MPPBA with $\Emc' = \varnothing$, hence, $\Amc'$ can be seen as a PPBA in this case.
\end{proof}

\subsection{Equivalence of $\varepsilon$-WPBA and $\varepsilon$-SPBA}

We now prove that the reset models $\varepsilon$-SPBA and $\varepsilon$-WPBA define the same class of $\omega$-languages. Furthermore, we can efficiently convert an $\varepsilon$-SPBA into an $\varepsilon$-WPBA and vice versa.

\begin{lemma}
\label{lem:SPBAtoWPBA}
Every $\varepsilon$-SPBA $\Amc$ is equivalent to an $\varepsilon$-WPBA $\Amc'$ that has the same number of states and uses one additional counter. If $\Amc$ is an SPBA, then $\Amc'$ is a WPBA.
\end{lemma}
\begin{proof}
Let $\Amc = (Q, \Sigma, q_0, \Delta, \Emc, F, C)$ be an $\varepsilon$-SPBA. We construct an equivalent $\varepsilon$-WPBA~$\Amc'$ that simulates $\Amc$, ensuring that no run visits an accepting state without resetting. 
To achieve that, we add an additional counter that tracks the number of visits of an accepting state (without resetting). 
Moreover, we modify~$C$ such that this new counter must be set to 1 when visiting an accepting state, thus disallowing to pass such a state without resetting.

We choose $\Amc' = (Q, \Sigma, q_0, \Delta', $\Emc'$ F, C')$ where
\begin{align*}
\Delta' =&\ \{(q, x, \vbf \cdot 0, q') \mid (q, x, \vbf, q') \in \Delta, q' \notin F\} \\
\cup&\ \{(q, x, \vbf \cdot 1, q_f) \mid (q, x, \vbf, q_f) \in \Delta, q_f \in F\},
\end{align*}

and, similarly,
\begin{align*}
\Emc' =&\ \{(q, \varepsilon, \vbf \cdot 0, q') \mid (q, \varepsilon, \vbf, q') \in \Emc, q' \notin F\} \\
\cup&\ \{(q, \varepsilon, \vbf \cdot 1, q_f) \mid (q, \varepsilon, \vbf, q_f) \in \Emc, q_f \in F\}.
\end{align*}

Finally, let $C' = C \cdot \{1\}$. We claim that $\Amc'$ is an $\varepsilon$-WPBA equivalent to $\Amc$.

Let $\alpha \in S_\omega(\Amc)$. Since $\Amc$ resets every time when visiting an accepting state, $\Amc'$ can simulate an accepting run using the same states and reset positions. 
In particular, the new counter will be 1 on every visit of an accepting state. 
Thus, the choice of $C'$ implies that $\alpha \in W_\omega(\Amc')$.

Now, let $\alpha \in W_\omega(\Amc')$. Due to the choice of $C'$, it is indispensable for $\Amc$ to reset the counters every time an accepting state is visited. 
Otherwise, the new counter tracking the number of visits of accepting states would be greater than 1, thus violating the weak reset-acceptance condition. 
Hence $\alpha \in S_\omega(\Amc)$.
\end{proof}

\begin{lemma}
\label{lem:WPBAtoSPBA}
Every $\varepsilon$-WPBA $\Amc$ is equivalent to an $\varepsilon$-SPBA $\Amc'$ with at most twice the number of states and the same number of counters. If $\Amc$ is an SPBA, then $\Amc'$ is a WPBA.
\end{lemma}
\begin{proof}
Let $\Amc = (Q, \Sigma, q_0, \Delta, \Emc, F, C)$ be an \mbox{$\varepsilon$-WPBA}. We construct an equivalent $\varepsilon$-SPBA~$\Amc'$ that simulates $\Amc$ by having the option to ``avoid" accepting states arbitrarily long. For this purpose, we create a non-accepting copy of $F$. Consequently, $\Amc'$ can decide to continue or reset a partial run using non-determinism.

We choose $\Amc' = (Q \cup \{\hat{q}_f \mid q_f \in F\}, \Sigma, q_0,  \Delta', \Emc', F, C)$, where
\begin{align*}
\Delta' = \Delta\ \cup&\ \{(q, x, \vbf, \hat{q}_f) \mid (q, x, \vbf, q_f) \in \Delta, q_f \in F\} \\
\cup&\ \{(\hat{q}_f, x, \vbf, q) \mid (q_f, x, \vbf, q) \in \Delta, q_f \in F\} \\
\cup&\ \{(\hat{q}_f, x, \vbf, \hat{q}_f') \mid (q_f, x, \vbf, q_f') \in \Delta, q_f, q_f' \in F\},
\end{align*}
and similarly,
\begin{align*}
\Emc' = \Emc\ \cup&\ \{(q, \varepsilon, \vbf, \hat{q}_f) \mid (q, \varepsilon, \vbf, q_f) \in \Emc, q_f \in F\} \\
\cup&\ \{(\hat{q}_f, \varepsilon, \vbf, q) \mid (q_f, \varepsilon, \vbf, q) \in \Emc, q_f \in F\} \\
\cup&\ \{(\hat{q}_f, \varepsilon, \vbf, \hat{q}_f') \mid (q_f, \varepsilon, \vbf, q_f') \in \Emc, q_f, q_f' \in F\}.
\end{align*}
We claim that $\Amc'$ is an $\varepsilon$-SPBA equivalent to $\Amc$.

$\Rightarrow$ We first show $W_\omega(\Amc)\subseteq S_\omega(\Amc')$. Let $\alpha \in W_\omega(\Amc)$ and $r = r_1 r_2 r_3\dots$ be an accepting run of $\Amc$ on $\alpha$ with reset positions $k_0 < k_1 < k_2 \dots$. Now $\Amc'$ is able to simulate $r$ by choosing the state $\hat{q}_i$ for every state $q_i \in F$ where $i \neq k_j, j \geq 1$. Hence, $\Amc'$ visits an accepting state if and only if $\Amc$ resets its counters at the same position. Thus, $\alpha \in S_\omega(\Amc')$.

$\Leftarrow$ To see that $S_\omega(\Amc')\subseteq W_\omega(\Amc)$, let $\alpha \in S_\omega(\Amc')$ and $r' = r_1' r_2' r_3' \dots$, with $r_i' = (p_{i-1}', \gamma_i, \vbf_i, p_i')$, where $p_i' \in \{p_i, \hat{p}_i\}$ for all $i \geq 0$, be a run of $\Amc'$ on $\alpha$ satisfying the strong reset-acceptance condition. 
Then $r = r_1 r_2 r_3\dots$ with $r_i = (p_{i-1}, \gamma_i, \vbf_i, p_i)$ is a run of $\Amc$ on $\alpha$. 
Furthermore, $r$ satisfies the weak reset-acceptance condition: let $k_1, k_2, \dots$ denote all positions in $r'$ where $p'_{f_i}$ is accepting. In particular, $k_1, k_2, \dots$ are an infinite number of (possible) reset positions, thus satisfying the weak reset-acceptance condition. Therefore, $\alpha \in W_\omega(\Amc)$.
\end{proof}

As a result, we call an $\omega$-language $L$ \emph{Reset-recognizable} if there is an SPBA $\Amc$ such that $S_\omega(\Amc) = L$. We write $\LReset$ to denote the class of all Reset-recognizable languages. Similarly, we call $L$ \emph{Prefix-recognizable} if there is a PPBA $\Amc$, such that $P_\omega(\Amc) = L$ and denote the class of Prefix-recognizable languages by \LPPBA.

%% file: epsilon.tex
\section{Schnepsilon-elimination}\label{sec:epsilon-elimination}

\subsection{$\epsilon$-elimination for PPBA}

We now show that $\epsilon$-transitions in $\epsilon$-PPBA can be eliminated, that is, every $\epsilon$-PPBA~$\Amc$ is equivalent to a PPBA~$\Amc'$. 
We proceed in two steps. In the first step, we show that we can convert every $\varepsilon$-PPBA into an equivalent $\varepsilon$-PPBA where all occurring $\varepsilon$-transitions are $\varepsilon$-loops. 
In the second step we show how to remove all $\varepsilon$-loops, thus obtaining a PPBA without $\varepsilon$-transitions.

Let $\Sigma = \{\sigma_1, \dots, \sigma_k\}$. The \emph{Parikh image} of a (finite) word $w \in \Sigma^*$ is the vector $p(w) = (|w|_{\sigma_1}, \dots, |w|_{\sigma_k}) \in \Nbb^k$. The definition extends to languages $L \subseteq \Sigma^*$ in the natural way: $p(L) = \{p(w) \mid w \in L\}$.
An important ingredient of our proof is (the first statement of) Parikh's Theorem \cite{parikh1966context}, stating the following.
\begin{theorem}[Parikh]
For every regular language $L$ the set $p(L)$ is semi-linear.
\end{theorem}

Let $\Gamma\subseteq \Nbb^k$ be finite. For a word $w=w_1\ldots w_\ell\in \Gamma^*$ we write $\sum w$ for $\sum_{1\leq i\leq k}w_i\in \Nbb^k$. 
For a language $L\subseteq \Gamma^*$ we define $\sum L=\{\sum w \mid w\in L\}$. Let $\Sigma=\{\sigma_1, \ldots, \sigma_k\}$ be an alphabet. 
Let $\mathcal{X}_\Sigma=\{\ebf^k_i \mid 1\leq i\leq k\}\subseteq \Nbb^k$ be the set of $k$-dimensional unit vectors. 
Then Parikh's Theorem implies that for every regular language $L$ over $\mathcal{X}_\Sigma$ the set $\sum L$ is semi-linear. 
Klaedtke and Ruess showed that this statement is true for arbitrary finite alphabets $\Gamma \subseteq \Nbb^k$. 

\begin{lemma}[Klaedkte, Ruess, Lemma 5 of \cite{klaedtkeruess}, rephrased]
\label{lemma:parikh}
Let $\Gamma \subseteq \Nbb^k$ be a finite alphabet. Then for every regular language $L$ over $\Gamma$ the set $\sum L$ is semi-linear. 
\end{lemma}


For an $\varepsilon$-PPBA $\Amc = (Q, \Sigma, q_0, \Delta, \Emc, F, C)$ we define $\Amc^\varepsilon$ to be the defined as $\Amc$ where all non-$\varepsilon$-transitions are removed, that is, $\Amc^\varepsilon = (Q, \Sigma, q_0, \varnothing, \Emc, F, C)$. 
Furthermore, for $p,q\in Q$ let~$\Amc_{p, q}$ be defined as $\Amc$ with initial state $p$ and a single accepting state $q$, \ie, $\Amc_{p,q} = (Q, \Sigma, p, \Delta, \Emc, \{q\}, C)$. Thus $\Amc^\varepsilon_{p, q} = (Q, \Sigma, p, \varnothing, \Emc, \{q\}, C)$. 
Let $\Bmc_{p,q}$ the NFA over the alphabet $\Gamma \subseteq \Nbb^k$ obtained from $\Amc_{p,q}^\varepsilon$ by replacing every labeled $\varepsilon$-transition with a transition where only the vector remains, that is, $\Bmc_{p,q} = (Q, \Gamma, \{p\}, \Delta', \{q\})$ where $\Delta' = \{(q_1, \vbf, q_2) \mid (q_1, \varepsilon, \vbf, q_2) \in \Emc\}$. 
Then, the following is immediate from \Cref{lemma:parikh}.
\begin{corollary}\label{crl:eps-semi-linear}
For every $\varepsilon$-PPBA $\Amc = (Q, \Sigma, q_0, \Delta, \Emc, F, C)$ and all $p,q\in Q$, the set $\sum L(\Bmc_{p, q})$ is semi-linear.
\end{corollary}

We first show that we may assume that the initial state has no non-loop $\epsilon$-transitions.

\begin{lemma}\label{lem:initial-state-only-loops}
For every $\epsilon$-PPBA $\Amc = (Q, \Sigma, q_0, \Delta, \Emc, F, C)$ there is an equivalent $\epsilon$-PPBA $\Amc'= (Q\cup \{q_0'\}, \Sigma, q_0', \Delta', \Emc', F, C')$ such that all outgoing $\epsilon$-transitions of $q_0'$ are $\epsilon$-loops.
\end{lemma}
\begin{proof}
We introduce a fresh initial state $q_0'$. 
The idea of the proof is as follows. For every $q\in Q$ reachable from $q_0$ by a sequence of $\epsilon$-transitions followed by a single transition labeled with a symbol $a$, we introduce a shortcut labeled with $a$. To account for the missing $\epsilon$-transitions we introduce $\epsilon$-loops on $q_0'$, so that the run up to $q$ can be simulated by the $\epsilon$-transitions on $q_0'$ followed by the $a$-shortcut to $q$. The labels of the new $\epsilon$-transitions on $q_0'$ are given by the automata $\Bmc_{q_0,q}$. Furthermore, we will use additional counters to make sure that we do not mix the transitions for different states $q,q'$. We come to the formal details. 

For every state $q \in Q$ we consider the set $\sum L(\Bmc_{q_0, q})$, which is semi-linear by \Cref{crl:eps-semi-linear}. 
Note that $\sum L(\Bmc_{q_0, q})$ is empty if $q$ is not reachable from $q_0$ via $\varepsilon$-transitions. 
Otherwise $\sum L(\Bmc_{q_0, q})$ can be written as the finite union of linear sets, say $C^{(q)}_1 \cup \dots \cup C^{(q)}_{\ell_q}$ where $C^{(q)}_i = \{b_{i, 0}^{(q)} + b_{i, 1}^{(q)}z_1 + \dots + b^{(q)}_{i, k_{i,q}} z_{k_{i,q}} \mid z_1, \dots, z_{k_{i,q}} \in \Nbb\}$. 
We denote by $k_{i, q}$ the number of period vectors of $C^{(q)}_i$, by $\ell_q$ the number of linear sets of the semi-linear set $\sum L(\Bmc_{q_0, q})$, and define $M = \sum_{q \in Q} \ell_q$. 
We fix an arbitrary bijection $\pi$ between the $C^{(q)}_i$ and $\{1, \dots, M\}$. 

We now introduce the new shortcuts and $\epsilon$-loops on $q_0'$. Informally, whenever a run of $\Amc$ starts with a sequence of $\varepsilon$-transitions $r_1 \dots r_m r_{m+1}$, where $r_1 \dots r_m \in \Emc^*$ and $r_{m+1} = (p_m, \alpha_1, \vbf_m, p_{m+1}) \in \Delta$, the automaton $\Amc'$ guesses the last state $p_m$ of the sequence of $\varepsilon$-transitions and take a shortcut from $q'_0$ to $p_{m+1}$. 
As $\rho(r_1 \dots r_m) \in C^{(p_m)}_i$ for some $i \leq \ell_{p_m}$, the automaton $\Amc'$ also guesses the set $C^{(p_m)}_i$ and use the $\varepsilon$-loops on $q_0'$ labeled with the period vectors of the $C^{(q)}_i$ concatenated with a $0$-$1$-vector of dimension~$M$ with a single $1$-entry. 
These loops replace all sequences of $\varepsilon$-transitions leaving $q_0$ and the new counters ensure that we do not mix the period vectors of different linear sets.

Formally, we define $\Amc'$ as follows.
\begin{align*}
\Delta' =&\ \{(p, a, \vbf \cdot \0^M, q) \mid (p,a,\vbf,q) \in \Delta\} \\
     \cup&\ \{(q_0', a, \vbf \cdot \0^M, q) \mid (q_0,a,\vbf,q) \in \Delta\} \\
     \cup&\ \{(q_0', a, (\vbf + b^{(p)}_{i,0}) \cdot \ebf^M_{\pi(C^{(p)}_i)}, q) \mid (p, a, \vbf, q) \in \Delta, 1 \leq i \leq \ell_p\},
\end{align*}
\begin{align*}
\Emc'   =&\ \{(p, \varepsilon, \vbf \cdot \0^M, q) \mid (p, \varepsilon, \vbf, q) \in \Emc, p \neq q_0\} \\
     \cup&\ \{(q'_0, \varepsilon, b^{(p)}_{i,j} \cdot \ebf^M_{\pi(C^{(p)}_i)}, q'_0) \mid p \in Q, 1 \leq i \leq \ell_p, 1 \leq j \leq k_{i,q}\},
\end{align*}
and
\begin{align*}
    C' = C \cdot \{z \ebf_i^M \mid z \in \Nbb, 1 \leq i \leq M\}.
\end{align*}

We prove that $\Amc'$ is equivalent to $\Amc$.

\smallskip
$\Rightarrow$ 
To show that  $P_\omega(\Amc)\subseteq P_\omega(\Amc')$, let $\alpha \in P_\omega(\Amc)$ with accepting run $r = r_1 r_2 r_3 \dots$ where $r_i = (p_{i-1}, \gamma_i, \vbf_i, p_i)$. We distinguish two cases.

In the first case assume that $r_1 = (q_0, \alpha_1, \vbf, p_1) \in \Delta$. In this case we can simply replace $r_1$ by $r_1' = (q_0', \alpha_1, \vbf \cdot \0, p_1)$ and continue the run as in $\Amc$, padding all vectors with zeros. 
That is, for $i \geq 2$ let $r_i' = (p_{i-1}, \gamma_i, \vbf_i \cdot \0^M, p_i)$. Then $r' = r_1' r_2' r_3'$ is a run of $\Amc'$ on $\alpha$. We show that $r'$ is also accepting. As $r$ is accepting, there are infinitely many $i$ such that $\rho(r_1 \dots r_i) \in C$. 
As $r'$ is basically equal to $r$ (up to padded zeros), there are infinitely many $i$ such that $\rho(r_1' \dots r_i') \in C \cdot \{\0^M\} \subseteq C'$, hence $r'$ is accepting.

For the second case assume that $r_1 \in \Emc$, \ie, the first transition of $\Amc$ is an \mbox{$\varepsilon$-transition}.
Let $r_m$ be the last occurrence of the initial $\varepsilon$-sequence, that is, $r_1 \dots r_m \in \Emc^*$ and $r_{m+1} = (p_m, \alpha_1, \vbf_m, p_{m+1}) \in \Delta$. 
Observe that $\rho(r_1 \dots r_m) \in C^{(p_m)}_i$ for some $i \leq \ell_{p_m}$, hence $\rho(r_1 \dots r_m) = b_{i, 0}^{(p_m)} + b_{i, 1}^{(p_m)}z_1 + \dots + b^{(p_m)}_{i, k_{i,p_m}} z_{k_{i,p_m}}$ for some $z_1, \dots z_{k_{i,p_m}}$. 
As $q_0'$ is equipped with $\varepsilon$-loops labeled with the $b_{i, j}^{(p_m)}$ (concatenated with $\ebf^M_{\pi(C^{(p_m)}_i)}$), $\Amc'$ can sum up $\rho(r_1 \dots r_m) - b_{i, 0}^{(p_m)}$ (with the additional counters) by taking $z_1$ times the transition $(q_0', \varepsilon, b^{(p_m)}_{i,1} \cdot \ebf^M_{\pi(C^{(p_m)}_i)}, q_0')$, $z_2$ times the transition $(q_0', \varepsilon, b^{(p_m)}_{i,2} \cdot \ebf^M_{\pi(C^{(p_m)}_i)}, q_0')$, and so on. 
Finally, we take the shortcut transition $\hat{r}_n = (q_0', a, (\vbf + b^{(p_m)}_{i,0}) \cdot \ebf^M_{\pi(C^{(p_m)}_i)}, p_{m+1})$ to add the missing base vector $b_{i, 0}^{(p_m)}$. From there we can continue the run as in $\Amc'$ again using the zero-padded transitions. 
This yields a run $r' = \hat{r}_1 \dots \hat{r}_{n-1}\hat{r}_n r'_{m+1}r'_{m+2} \dots$ of $\Amc'$ on~$\alpha$ where $\hat{r}_1 \dots \hat{r}_{n-1}$ is a sequence of $\varepsilon$-loops, $\hat{r}_n$ is the shortcut transition, and $r'_i$ for $i \geq m+1$ is defined as in the first case. 
It remains to show that $r'$ is accepting. First observe that $\rho(\hat{r}_1 \dots \hat{r}_n) \in \{\rho(r_1 \dots r_{m+1})\} \cdot \{z \ebf_{\pi(C^{(p)}_i)}^M \mid z \in \Nbb\}$ and for all $j \geq m+1$ we have \smash{$\rho(r'_{m+1} \cdot r'_j) = \vbf \cdot \0^M$} for some $\vbf$. Thus, all (infinitely many) accepting hits of $r$ that occur after $r_m$ translate into accepting hits of $r'$, hence $r'$ is accepting. \hfill $\lrcorner$

\medskip
$\Leftarrow$ To show vice versa that $P_\omega(\Amc')\subseteq P_\omega(\Amc)$, let $\alpha \in P_\omega(\Amc')$ with accepting run $r' = r'_1 r'_2 r'_3 \dots$ where $r'_i = (p_{i-1}, \gamma_i, \vbf_i \cdot \ubf_i, p_i)$ for $\ubf_i \in \{z \ebf_i^M \mid z \in \{0,1\}, 1 \leq i \leq M\}$. Again we distinguish two cases.

In the first case assume that $r'_1$ is of the form $(q_0', a, \vbf \cdot \0^M, p_1)$. By construction we have $p_1 \neq q_0'$. 
In particular, we have $p_i \neq q_0'$ for all $i \geq 1$. Hence, we can replace $r'_1$ by $r_1 = (q_0, a, \vbf, p_1)$ and all $r'_i$ by $r_i = (p_{i-1}, a, \vbf_i, p_i)$ for $i \geq 2$. 
Then the run $r_1 r_2 r_3 \dots$ is an accepting run of $\Amc$ on $\alpha$.

In the second case assume that $r'_1$ is not of the form mentioned above. Then there is a (unique) $m \geq 1$ such that $r'_m$ is of the form $(q_0', \alpha_1, \vbf \cdot \ebf_j^M, p_m)$ for some $j \leq M$.
Let $C^{(q)}_i$ be the set with $\pi(C^{(q)}_i) = j$ (which is uniquely determined as $\pi$ is a bijection). 
In particular, for all $k \geq 1$ we have $\rho(r'_1 \dots r'_k) = \ubf \cdot z\ebf_j^M$ for some $z \in \Nbb$ by the choice of $C'$. 
Hence, we can replace $r'_1 \dots r'_m$ by a sequence of $\varepsilon$-transitions $r_1 \dots r_n$ in $\Amc$ such that $r_1 = (q_0, \varepsilon, \vbf_1, p_1)$, $r_n = (p_{n-1}, \varepsilon, \vbf_n, q)$ and $\rho(r_1 \dots r_n) \in C^{(q)}_i$. Let $r_{n+1} = (q,\alpha_1, \vbf, p_m)$ where $\vbf$ is the vector on $r'_m$. 
Note that $r_{n+1} \in \Delta$ by construction. 
Then $r_1 \dots r_n r_{n+1} r_{m+1} r_{m+2}\ldots$ where $r_i = (p_{i-1}, \gamma_k, \vbf_i, p_i)$ for all $i \geq m+1$ is a valid run of $\Amc$ on $\alpha$. 
Furthermore, if $\rho(r_1',\ldots, r_i')=\ubf\cdot z\ebf_j^M$ for some $i\geq m+1$, then $\rho(r_1,\ldots, r_{n+1}r_{m+1},\ldots, r_i)=\ubf$. 
Since $r'$ is accepting we conclude that $r$ is accepting. 
\end{proof}

\begin{lemma}
\label{lemma:loops}
Let $\Amc = (Q, \Sigma, q_0, \Delta, \Emc, F, C)$ be an $\varepsilon$-PPBA of dimension $d$. There is an equivalent $\varepsilon$-PPBA~$\Amc'$ where all $\varepsilon$-transitions are $\varepsilon$-loops.
\end{lemma}
\begin{proof}
By \Cref{lemma:GPPBA} we may assume that $\Amc$ has only a single accepting state $f$, that is, $F = \{f\}$, and by \Cref{lem:initial-state-only-loops} we may assume that all outgoing $\epsilon$-transitions of the initial state are $\epsilon$-loops. 

In the beginning, $\Amc' = (Q', \Sigma, q_0', \Delta', \Emc', F', C')$ is a copy of $\Amc$, which will be modified step-by-step. First, we remove all non-loop $\varepsilon$-transitions from $\Amc'$.

The intuition for the construction is as follows. 
We introduce new states such that in every run of $\Amc$ on an infinite word we can replace every maximal partial run $r_i \dots r_j$ of \mbox{$\varepsilon$-transitions}, where $r_{i-1} = (p_{i-2}, \alpha_{z-1}, \vbf_{i-1}, p_{i-1}), r_i = (p_{i-1}, \varepsilon, \vbf_i, p_i), r_j = (p_{j-1}, \varepsilon, \vbf_j, p_j)$, and $r_{j+1} = (p_j, \alpha_{z}, \vbf_{j+1}, p_{j+1})$, by a sequence of $\varepsilon$-loops on a new state depending on $p_{i-1}, p_{j}$, and $\sum L(\Bmc_{p_{i-1}, p_j})$. 
The (finite) partial run $r_i \dots r_j$ translates into an accepting run in $\Bmc_{p_{i-1}, p_j}$, thus $\rho(r_i \dots r_j) \in \sum L(\Bmc_{p_{i-1}, p_j})$. 
By \Cref{crl:eps-semi-linear} the set $\sum L(\Bmc_{p_{i-1}, p_j})$ is semi-linear and hence a finite union of linear sets, each defined by a base vector and a set of period vectors, we can encode the period vectors as $\varepsilon$-loops and shortcut the base vector. 
Some care must be taken if we visit the accepting state $f$ in the partial run $r_i \dots r_j$.

Let us continue with the formal construction. 
We iterate through all pairs $p,q$ of states. If $L(\Bmc_{p, q}) = \varnothing$, we do nothing and continue with the next pair of states. 
Otherwise, the state $q$ is reachable from $p$ in $\Bmc_{p, q}$ and $\sum L(\Bmc_{p, q})$ is non-empty and can be written as the finite union of linear sets. 
For the sake of readability we assume that $\sum L(\Bmc_{p, q}) = C_{p,q}$ is a single linear set. 
For the general case we apply the following construction independently for each of the linear sets in the union. In the following, for any pair of states $s,t$ we denote by $b_{s,t}$ the base vector of $C_{s,t}$. 

We add a new state $(p,q)$ to the new state set $Q'$ (the new state is the tuple of states~$p$ and $q$). 
If $C_{p,q} = \{b_{p,q} + b_1z_1 + \dots + b_kz_k \mid z_1, \dots, z_k \in \Nbb\}$ 
we add $\varepsilon$-loops labeled with the period vectors $b_1, \dots, b_k$ of $C_{p,q}$ to the state $(p,q)$. 
Finally, we add \emph{shortcuts} to $(p,q)$, that is, for every transition of the form $(s, a, \vbf, p) \in \Delta$ we add a transition $(s, a, \vbf + b_{p,q}, (p,q))$ to $\Delta'$ (note that we add the base vector $b_{p,q}$ of $C_{p,q}$ to $\vbf$). 
Likewise, for every transition of the form $(q, a, \vbf, s) \in \Delta$ we add a transition $((p,q), a, \vbf, s)$ to $\Delta'$. 
Finally, we connect the new states by further adding transitions
$((p,q), a, \vbf+b_{s,t}, (s, t))$ to $\Delta'$ for every transition $(q, a, \vbf, s) \in \Delta$ and $t \in Q$ such that $(s,t)\in Q'$  (again assuming that $C_{s,t}$ is linear; otherwise, we would have multiple copies of $(s,t)$, each of which gets connected by a transition as above.).

	
	
	
	

After the exhaustive application of this procedure $\Amc'$ can simulate every run of $\Amc$ in the sense that all $\varepsilon$-transitions that appear in $\Amc$ have been reduced to a number of new states equipped with $\varepsilon$-loops. 
However, if the accepting state $f$ is reachable from~$p$ in $\Bmc_{p,q}$ we must take into account that the procedure above might have shortcut the accepting state, which leads to missing accepting hits. 
To prevent this, we need to add additional new states.


We iterate over all pairs of states $p, q \in Q$ again. 
However, this time we consider the automata $\Amc_{p,f}^\varepsilon$ and $\Amc_{f,q}^\varepsilon$. 
As above, let $\Bmc_{p,f}$ and $\Bmc_{f,q}$ be the NFA whose alphabets are subsets of $\Nbb^d$ obtained from $\Amc_{p,f}^\varepsilon$, resp.\ $\Amc_{f,q}^\varepsilon$ by keeping only the vectors of the $\varepsilon$-transitions. 
By \Cref{crl:eps-semi-linear} the sets $\sum L(\Bmc_{p, f})$ and $\sum L(\Bmc_{f, q})$ are semi-linear.

If both of these sets are non-empty, we can write each of them as a finite union of linear sets. 
Again, for the sake of readability, we assume that $\sum L(\Bmc_{p, f}) = C_{p,f}$ and $\sum L(\Bmc_{f, q}) = C_{f,q}$ are linear sets. For the general case, we apply the following construction for every combination of a linear set of $L(\Bmc_{p, f})$ with a linear set of $L(\Bmc_{f, q})$ independently.

Let $C_{p,f} = \{b_{p,f} + b_1z_1 + \dots + b_kz_k \mid z_1, \dots, z_k \in \Nbb\}$ and $C_{f,q} = \{b_{f,q} + c_1z_1 + \dots + c_\ell z_\ell \mid z_1, \dots, z_\ell\in\Nbb\}$. 
We add a new accepting state $(p,f,q)$ to $Q'$ and equip it with $\varepsilon$-loops labeled with the period vectors $b_1, \dots, b_k$ of $C_{p,f}$. 
Furthermore, for each transition of the form $(r, a, \vbf, p) \in \Delta$ we add a shortcut $(r, a, \vbf + b_{p,f}, (p,f,q))$ to $\Delta'$.

Let us now connect the $(p,q)$-states to the states just introduced. 
For every transition $(q, a, \vbf, s) \in \Delta$ we add the ingoing transitions $((p,q), a, \vbf+b_{s,f}, (s,f,t))$ for all $t \in Q$ to~$\Delta'$.

In the next step we introduce the outgoing transitions of the $(p,f,q)$-states. Unfortunately, this situation is more complicated and we have to introduce yet more states $(p,f,q,s)$ and $(p,f,q,s,t)$. 
The idea is that these states act like a copy of $s$ resp. $(s,t)$, but are additionally equipped with the $\varepsilon$-loops labeled with the period vectors of $C_{f,q}$. 
We cannot simply add these loops to $(p,f,q)$, as this might lead to accepting hits in $\Amc'$ that are not possible in $\Amc$ by using these loops (which are for vectors of $\varepsilon$-sequences \emph{leaving}~$f$). 
We cannot simply ignore them either, as they are necessary to simulate the runs of $\Amc$ appropriately. 
Hence the copies, which allow us to use the loops without generating false accepting hits.

Formally, we insert a new state $(p,f,q,s)$ for all $p,q,s \in Q$ and a new state $(p,f,q,s,t)$ for all $p,q\in Q$ and $(s,t)\in Q'$. For every transition $(q, a, \vbf, s) \in \Delta$ we add the following transitions to $\Delta'$:
$((p,f,q), a, \vbf+b_{f,q}, (p,f,q,s))$, and 
 $((p,f,q), a, \vbf+b_{f,q}+b_{s,t}, (p,f,q,s,t))$ for all $t \in Q$ such that there exists $(s,t)\in Q'$.

Finally, we connect the new states $(p,f,q,s)$ and $(p,f,q,s,t)$ as follows. First, $(p,f,q,s)$ and $(p,f,q,s,t)$ are equipped with $\varepsilon$-loops labeled with the period vectors of $C_{f,q}$. 
Furthermore, $(p,f,q,s)$ has all outgoing transitions of $s$, that is, for every transition $(s, a, \vbf,t) \in \Delta'$ we also add a transition $((p,f,q,s), a, \vbf, t)$ to $\Delta'$. 
Similarly, $(p,f,q,s,t)$ gets all outgoing transitions of $(s,t)$ in $\Delta'$. Additionally, $(p,f,q,s,t)$ gets all $\varepsilon$-loops of $(s,t)$.

\smallskip
Now, all remaining $\varepsilon$-transitions of $\Amc'$ are loops and we have finished the construction of $\Amc'$. $\hfill \lrcorner$

\bigskip

We prove that $\Amc'$ is equivalent to $\Amc$. In the following, by an $\epsilon$-sequence we mean a maximal sequence of $\epsilon$-transitions containing at least one non-loop $\epsilon$-transition.

\medskip

$\Rightarrow$ To show that $P_\omega(\Amc)\subseteq P_\omega(\Amc')$, let $\alpha \in P_\omega(\Amc)$ with accepting run $r = r_1 r_2 r_3 \dots$. If there are no $\varepsilon$-sequences in $r$ we are done (as $r$ is also a run of $\Amc'$ on $\alpha$).

Otherwise, we construct an accepting run $r'$ of $\Amc'$ on $\alpha$ by replacing $\epsilon$-sequences step-by-step. 
Let $i$ be minimal such that $r_i \dots r_j$ is an $\epsilon$-sequence. 
Note that because~$q_0$ has no non-loop $\epsilon$-transitions we have $i>1$. Let $r_{i-1} = (p_{i-2}, \alpha_{z-1}, \vbf_{i-1}, p_{i-1})$, $r_i = (p_{i-1}, \varepsilon, \vbf_i, p_i), r_j = (p_{j-1}, \varepsilon, \vbf_j, p_j)$, and $r_{j+1} = (p_j, \alpha_{z}, \vbf_{j+1}, p_{j+1})$. Similarly, let $k > j$ be minimal such that $r_k \dots r_\ell$ is an $\epsilon$-sequence, that is, $r_k \dots r_\ell$ is the second $\varepsilon$-sequence in $r$ (it might be the case that such $k$ does not exist, we handle this case explicitly below). Note that we have $\rho(r_i \dots r_j) \in \sum L(B_{p_{i-1}, p_j})$ and $\rho(r_k \dots r_\ell) \in \sum L(B_{p_{k-1}, p_\ell})$.

\smallskip
We distinguish the (combination of the) following cases. 
\begin{itemize}
    \item There is $f_1$ with $i \leq f_{1} \leq j$ such that we have an accepting hit in $r_{f_1}$ (F) or not~(N). Recall that we have only one accepting state $f$ and $f_1$ here denotes the position of the accepting hit.
    \item We have $k = j+2$, that is, there is just a single non-$\varepsilon$-transition between the two sequences $r_i \dots r_j$ and $r_k \dots r_\ell$ (1) or $k > j+2$, that is, there are at least two non-$\epsilon$-transitions between the two sequences (2).
    \item There is $f_2$ with $k \leq f_2 \leq \ell$ such that we have an accepting hit in $r_{f_2}$ (F) or not, or~$k$ does not even exist (N).
\end{itemize}
Hence, we consider eight cases in total.

\medskip
$\bullet$ Case (F1N). This means we need to take care of an accepting hit in the first \mbox{$\epsilon$-sequence} (at position $f_1$), but not in the second $\epsilon$-sequence (here we assume that the second sequence exists, the other case is handled in (F2N)), and there is just a single non-$\varepsilon$-transition between these sequences.

We replace $r_{i-1}$ by $\hat{r}_{i-1} = (p_{i-2}, \alpha_{z-1}, \vbf_{i-1} + b_{p_{i-1}, f}, (p_{i-1}, f, p_j)) \in \Delta'$. 
As there is an accepting hit at position $f_1$, we have $\rho(r_1 \dots r_{f_1}) = \rho(r_1 \dots r_{i-1}) + \rho(r_i \dots r_{f_1}) \in C$. In particular, we have $\rho(r_i \dots r_{f_1}) \in C_{p_{i-1}, p_{f_i}}$. 
By construction, there are $\varepsilon$-loops on $(p_{i-1}, f, p_j)$ labeled with the period vectors of $C_{p_{i-1}, p_{f_i}}$. 
Hence, we can replace $r_i \dots r_{f_1}$ by a sequence of $\varepsilon$-loops $\hat{r}_i \dots \hat{r}_m$ on $(p_{i-1}, f, p_j)$ with $\rho(\hat{r}_i \dots \hat{r}_m) = \rho(r_i \dots r_{f_1}) - b_{p_{i-1}, f}$. 
As $b_{p_{i-1}, f}$ has already been added to $\hat{r}_{i-1}$, the sequence $\hat{r} = r_1 \dots r_{i-2} \hat{r}_{i-1} \hat{r}_i \dots \hat{r}_m$ is a prefix of a run of $\Amc'$ on $\alpha$ with $\rho(\hat{r}) \in C$.

We describe how to continue the run at this point. 
We take the transition $\hat{r}_{m+1} = ((p_{i-1}, f, p_j), \alpha_z, \vbf_{j+1} + b_{f, p_j} + b_{p_{j+1}, p_\ell}, (p_{i-1}, f, p_j, p_{j+1}, p_\ell)) \in \Delta$ (recall that $p_{j+1} = p_{k-1}$). 
As $(p_{i-1}, f, p_j, p_{j+1}, p_\ell)$ is equipped with $\epsilon$-loops labeled with the period vectors of $C_{f,p_j}$ and also with the period vectors of $C_{j+1,\ell}$, we can replace the partial run $r_{m+1}\ldots r_j$ with a sequence of $\varepsilon$-loops $\hat{r}_{m+2} \dots \hat{r}_n$ on $(p_{i-1}, f, p_j,p_{j+1},p_\ell)$ with $\rho(\hat{r}_{f_1+1} \dots \hat{r}_j) = \rho(r_{f_1+1} \dots r_{j}) - b_{f, p_j}$ and the partial run $r_{j+2}\ldots r_\ell$ with a sequence of $\varepsilon$-loops $\hat{r}_{n+1} \dots \hat{r}_o$ on $(p_{i-1}, f, p_j,p_{j+1},p_\ell)$ with $\rho(\hat{r}_{j+2} \dots \hat{r}_\ell) = \rho(r_{j+2} \dots r_{\ell}) - b_{p_{j+1}, p_\ell}$. 
As $b_{p_{j+1}, p_\ell}$ and $b_{f, p_j}$ have already been added to $\hat{r}_{m+1}$, the partial run $\hat{r}_1 \dots \hat{r}_m \dots \hat{r}_n \dots \hat{r}_o$ is equivalent to $r_i \dots r_\ell$ in the sense that both runs have (at least) one accepting hit, read $\alpha_{z-1} \alpha_{z}$, have the same extended Parikh image, and "fit" into the whole run $r$, as $(p_{i-1}, f, p_j, p_{j+1}, p_l)$ has the same outgoing transitions (including possible shortcuts) as $p_l$.

We now continue with the next $\varepsilon$-sequence in $r$.

\medskip
$\bullet$ Case (F1F). We do the exact same replacement as in (F1N). We thereby lose an accepting hit in the second sequence, however, this is not a problem, as we still have an accepting hit in the first sequence and we have infinitely many accepting hits to come.

\medskip
$\bullet$ Case (F2N). This is similar to (F1N) but we chose $\hat{r}_{m+1} = ((p_{i-1}, f, p_j), \alpha_z, \vbf_{j+1} + b_{f, p_j}, (p_{i-1}, f, p_j, p_{j+1}))$. As we do not need to consider the set $C_{k, \ell}$ at this point, we just replace $r_{f_1 + 1} \dots r_j$ by a matching sequence $r_{m+2} \dots r_n$ of $\varepsilon$-loops on $(p_{i-1}, f, p_j, p_{j+1})$.

At this point, we continue as if we were in $p_{j+1}$. Note that we are done if the second sequence does not exist.

\medskip
$\bullet$ Case (F2F) is the same as (F2N).

\medskip
$\bullet$ Case (N1F). Here we replace $r_{i-1}$ by $\hat r_{i-1} = (p_{i-2}, \alpha_{z-1}, \vbf_i + b_{p_{i-1}, p_j}, (p_{i-1}, p_j)) \in \Delta'$. As $\rho(r_i \dots r_j) \in C_{p_{i-1}, p_j}$, we can replace this partial run by a sequence of $\varepsilon$-loops $\hat{r}_i \dots \hat{r}_m$ on $(p_{i-1}, p_j)$ with $\rho(\hat{r}_i \dots \hat{r}_m) = \rho(r_i \dots r_j) - b_{p_{i-1}, p_j}$. 
As~$b_{p_{i-1}, p_j}$ has already been added to $\hat{r}_{i-1}$, the sequence $\hat{r}_{i-1} \hat{r}_i \dots \hat{r}_m$ is equivalent to $r_{i-1} r_i \dots r_j$ in the sense that both runs read $\alpha_{z-1}$, have the same Parikh images and fit into the whole run $r$, as we can continue the run from $(p_i, p_j)$ in exactly the same way as in $p_j$, hence we continue with the next $\varepsilon$-sequence.

\medskip
$\bullet$ The remaining cases (N1N), (N2F) and (N2N) are the same as (N1F).

\medskip
All accepting hits outside of $\varepsilon$-sequences translate one-to-one. This finishes the proof of the forward direction. 
$\hfill \lrcorner$

\bigskip

 $\Leftarrow$ To show that $P_\omega(\Amc')\subseteq P_\omega(\Amc)$, let $\alpha \in P_\omega(\Amc')$ with accepting run $r' = r'_1 r'_2 r'_3 \dots$. 
 If all states that appear in $r$ belong to the state set $Q$ of $\Amc$, we are done as the run is also an accepting run of $\Amc$.
 
 Otherwise, we construct an accepting run $r$ of $\Amc$ on $\alpha$ step-by-step. Let $i$ be minimal such that $r'_i = (p'_{i-1}, \alpha_z, \vbf'_i, p'_i)$ contains a state that is not part of $Q$. 
 As all (accepting) runs of $\Amc'$ start in $q_0$ (which belongs to $Q$), we have that $p'_i \notin Q$. We distinguish two cases.
 
 If $p'_i = (s,t)$ for some $s,t \in Q$, we have $\vbf_i' = \vbf_i + b_{s,t}$ by the choice of $\Delta'$. 
 Let $j \geq i$ be maximal such that for all $i \leq k \leq j$ we have that~$r'_k$ is an $\varepsilon$-loop on $(s,t)$, \ie, of the form $((s,t), \varepsilon, \vbf_k', (s,t))$.
 
 By the semantics of $(s,t)$, there is an $\varepsilon$-sequence $r_i \dots r_n$ in $\Amc$ that starts in $s$ and ends in $t$. 
 To be precise, we have $r_i = (s, \varepsilon, \vbf_i, p_i)$ and $r_n = (p_{n-1}, \varepsilon, \vbf_n, t)$, and there is a transition $r_{i-1} = (p_{i-2}, \alpha_z, \vbf_i, s) \in \Delta$ such that $\vbf_i = \vbf_i'$ and hence $\rho(r_{i-1} r_i \dots r_n) = \rho(r'_i \dots r'_j)$. 
 Furthermore, observe that $p_{i-2} = p'_{i-1}$, hence $r_{i-1} r_{i} \dots r_n$ is equivalent to $r'_i \dots r'_j$ in the sense that both runs read $\alpha_z$, have the same extended Parikh images, and fit into the whole run $r'$ as we can continue the run from $t$ exactly the same way as in~$(s,t)$.
 
 If $p'_i = (s,f,t)$ for some $s,t \in Q$, we have $\vbf_i' = \vbf_i + b_{s,f}$ by the choice of $\Delta'$. 
 Let $j \geq i$ be maximal such that for all $i \leq k \leq j$ we have that $r'_k$ is an $\varepsilon$-loop on $(s,f,t)$, \ie, $r'_k$ is a transition of the form $((s,f,t), \varepsilon, \vbf_k', (s,f,t))$ and there is an accepting hit in $r_j'$ (if there is no accepting hit, let $j$ be the last $\varepsilon$-loop in this sequence).
 
 By the semantics of $(s,f,t)$ there is an $\varepsilon$-sequence $r_i \dots r_m \dots r_n$ in $\Amc$ that starts in $s$, visits $f$, and ends in $t$. 
 To be precise, we have $r_i = (s, \varepsilon, \vbf_i, p_i)$, $r_m = (p_{m-1}, \varepsilon, \vbf_m, f)$, $r_n = (p_{n-1}, \varepsilon, \vbf_n, t)$, and there is a transition $r_{i-1} = (p_{i-2}, \alpha_z, \vbf_i, s) \in \Delta$ such that $\vbf_i = \vbf_i'$. In particular, if there is an accepting hit in $r'_j$, we can choose $m$ such that $\rho(r_{i-1} r_i \dots r_m) = \rho(r_i' \dots r'_j)$, hence there is also an accepting hit in $r_m$. 
 By the choice of $\Delta'$, $r'_j$ is followed by a (possibly empty) sequence of $\varepsilon$-loops on $(s,f,t)$, followed by a transition of the form $\delta_1 = ((s,f,t), \alpha_{z+1}, \vbf + b_{f,t}, (s,f,t,p))$ or $\delta_2 = ((s,f,t), \alpha_{z+1}, \vbf + b_{f,t}+b_{p,q}, (s,f,t,p,q))$ for some $p,q \in Q$. 
 This sequence of $\varepsilon$-loops on $(s,f,t)$ matches to a sequence of \mbox{$\varepsilon$-transitions} in $\Amc$ that starts and ends in $f$ with the same extended Parikh image, hence we replace it accordingly.
 
 We now consider the next transition. In the first case we assume that it is of the form~$\delta_1$. 
 Recall that $(s,f,t,p)$ is equipped with several $\varepsilon$-loops labeled with the period vectors of~$C_{f,t}$, as well with possible $\varepsilon$-loops of $p$ in $\Amc$. 
 This means that the transition is followed by a possibly empty sequence of $\varepsilon$-loops on $(s,f,t,p)$. 
 Without loss of generality we assume that they are ordered in such a way that first all $\varepsilon$-loops labeled with period vectors appear, say $\hat{r}_1 \dots \hat{r}_o$, followed by possible $\varepsilon$-loops of $p$ (this is not a problem, as we just swap $\varepsilon$-loops and do not need to take care of any accepting hits as $(s,f,t,p)$ is non-accepting). 
 By construction we have $\rho(\hat{r}_1 \dots \hat{r}_o) = \rho(r_{m+1} \dots r_n) - b_{f,t}
 $. As $b_{f,t}$ has already been added to $\delta_1$, we replace $\delta_1 \hat{r_1} \dots \hat{r}_o$ by $r_{m+1} \dots r_n r_{n+1}$ where $r_{n+1} = (t, \alpha_{z+1}, \vbf, p)$. As $p$ is equipped with all outgoing transitions of $(s,f,t,p)$ (including possible $\varepsilon$-loops, but no other $\varepsilon$-transitions), we can continue the run in $\Amc$ the same way as in $\Amc'$.
 
 Finally, we assume that the next transition is of the form $\delta_2$. This case is similar, with the only exception that $(s,f,t,p,q)$ "behaves" like $(p,q)$ in the sense that $(s,f,t,p,q)$ has all outgoing transitions of $(p,q)$. Thus, we continue from here as in the first case, where we handle states of this form. 
 
 Again, all accepting hits using states in $Q$ translate one-to-one. This finishes the proof of the backward direction. 
\end{proof}

We now proceed to eliminate the remaining $\epsilon$-loops. We need the following lemma for $\epsilon$-elimination for automata on finite words. 

\begin{lemma}\label{lem:KR-finite}
[Theorem 22 of Klaedtke and Ruess~\cite{klaedtkeruess}, reformulated]
For every PA $\Amc = (Q, \Sigma, q_0, \Delta, F, C)$ of dimension $d$ \emph{on finite words} (with $\epsilon$-transitions) there exists an equivalent $\epsilon$-free PA $\hat\Amc = (Q, \Sigma, q_0, \hat\Delta, F, \hat C)$ on the same state set of dimension $d + |Q| -1$.
\end{lemma}

\begin{lemma}
Let $\Amc = (Q, \Sigma, q_0, \Delta, \Emc, F, C)$ be an $\varepsilon$-PPBA of dimension $d$ where all occuring $\varepsilon$-transitions are $\varepsilon$-loops. There is an equivalent PPBA $\Amc'$.
\end{lemma}
\begin{proof}
Let us first sketch the proof idea. Intuitively, we will split the automaton into two parts. We will guess the set $S$ of states that will be seen infinitely often. In the first part of the automaton we will deal with the set of states that are seen only finitely often. For this, we apply the construction of 
\Cref{lem:KR-finite} and make all accepting states non-accepting to obtain the automaton $\hat \Amc$. 
We non-deterministically switch to the second part, where we will verify that for some set $S\subseteq Q$ exactly the states of $S$ will be seen infinitely often.
We follow the idea of Klaedtke and Ruess~\cite{klaedtkeruess} for finite words:  
Since it does not matter when and in what order vectors are added, we can simulate $\epsilon$-loops by an appropriate modification of the semi-linear set. 
Instead of $\epsilon$-looping on a state we can intuitively ``substract'' the semi-linear set corresponding to the loop from $C$. 
Formally, we will construct an automaton $\Amc_S$ for each possible guess of $S$ and shortcut appropriately from $\hat \Amc$. We will work with one semi-linear set for each $\Amc_S$, so that we formally construct an MPPBA. 
We conclude by applying \Cref{lemma:GPPBA} to translate this automaton finally to an equivalent PPBA. Let us come to the formal proof. 

\medskip

By \cref{lemma:GPPBA} we may assume that $\Amc$ has only a single accepting state, say $F = \{f\}$. We first interpret $\Amc$ as a PA on finite words and denote by $\hat\Amc$ the $\epsilon$-free automaton obtained from \Cref{lem:KR-finite} by
padding every vector with $d$-many zeros (that is $\hat\Amc$ is of dimension $2d + |Q| - 1$ instead of $d + |Q| - 1$). We denote the semi-linear set of $\hat \Amc$ by $\hat C$. 

Now, for every non-empty subset of states $S \subseteq Q$ containing at least $f$, we construct an MPPBA $\Amc_S$ of dimension $d$, as follows. 
Let the states of $S$ be ordered arbitrarily, say $S = \{s_1, \dots, s_m\}$. $S$ is a candidate set for the set consisting exactly of those states that will be visited infinitely often. 
We connect~$\hat\Amc$ to all $\Amc_S$ by shortcutting all transitions in~$\hat\Amc$ that lead to~$f$ to the initial states of the $\Amc_S$. In this way we can non-deterministically switch from~$\hat\Amc$ to some $\Amc_S$. 

The automaton $\Amc_S$ consists of $m+1$ copies of $\Amc$ as well as a fresh state $q_S$, which is the initial state of $\Amc_S$, as well as the only accepting state of $\Amc_S$. 
In the following, we call the $i$th copy of $\Amc_S$ the $i$th \emph{layer} of $\Amc_S$. 
By allowing to switch from the $i$th layer to the $(i+1)$st layer only after visiting state $s_i$, the $m$ layers ensure that we visit all of the states $s_1, \dots, s_m$ infinitely often. 
From the last layer we have shortcuts into the new state~$q_S$, allowing us to "wait" for the point where the next transition would bring us to~$f$ and the counters to a value in the semi-linear set. 
Additionally, $q_S$ has basically the same outgoing transitions as the accepting state $f$, but they lead into the layer for $s_1$.

Formally, let $\Amc_S = ((S \times \{1, \dots, m+1\}) \cup \{q_S\}, \Sigma, q_S, \Delta_S, \{q_S\}, C_S)$, where
\begin{align*}
\Delta_S =\ &\{((s_i, k), a, \0^{d+|Q|-1} \cdot \vbf, (s_j,k)) \mid (s_i, a, \vbf, s_j) \in \Delta, k \leq m+1\} \\ 
      \cup\ &\{((s_i, i), a, \0^{d+|Q|-1} \cdot \vbf, (s_j,i+1) \mid (s_i, a, \vbf, s_j) \in \Delta, i \leq m\} \\ 
      \cup\ &\{((s_i, m+1), a, \0^{d+|Q|-1} \cdot \vbf, q_S) \mid (s_i, a, \vbf, f) \in \Delta\} \\ 
      \cup\ &\{(q_S, a, \0^{d+|Q|-1} \cdot \vbf, (s_i, 1) \mid (f, a, \vbf, s_i) \in \Delta\}. 
\end{align*}

For two semi-linear sets $C_1, C_2$ of dimension $d$, let $C_1 - C_2 = \{\ubf \in \Nbb^d \mid \ubf + \vbf \in C_1$ $\text{for some }\vbf \in C_2\}$. As shown by Klaedtke and Ruess, the set $C_1 - C_2$ remains semi-linear (observe that it is definable in Presburger Arithmetic). 
For $q \in Q$, let $C_q = \sum L(B_{q, q})$ and choose $C_S(q_S) = \{\0^{d+|Q|-1}\} \cdot (\vec{C} - C_{s_1} - \dots - C_{s_m})$, where $\vec{C}$ is defined as $C$ but without any base vectors.

Finally, let $\Amc' = (Q', \Sigma, q_0, \Delta', \{q_S \mid S \subseteq Q\}, C')$ be the (disjoint) union of $\hat\Amc$ and all~$\Amc_S$ with additional transitions $(q, a, \vbf, q_S)$ for each $(q, a, \vbf, f) \in \Delta$ and $C'(q_S) = \hat{C} + C_S(q_S)$. We claim that $\Amc'$ is equivalent to $\Amc$.

$\Rightarrow$ We first show $P_\omega(\Amc)\subseteq P_\omega(\Amc')$. Let $\alpha \in P_\omega(\Amc)$ with accepting run $r = r_1 r_2 r_3 \dots$ where $r_i = (p_{i-1}, \gamma_i, \vbf_i, p_i)$. 
Let~$S$ be the set of states that appear infinitely often in $r$ and let $i$ be minimal such that there is an accepting hit in $r_i$ and we have $p_j \in S$ for all $j \geq i$. 
Let $u = \alpha_1 \dots \alpha_z \in \Sigma^*$ be the prefix of $\alpha$ that has been read upon visiting $p_i$, and let $\beta = \alpha_{z+1} \alpha_{z+2} \dots \in \Sigma^\omega$, \ie, $\alpha = u\beta$.

First observe that $u \in L(\Amc)$ as $r_1 \dots r_i$ is an accepting run of $\Amc$ on $u$ by definition. 
As a consequence of \Cref{lem:KR-finite}, the automaton $\hat\Amc$ and $\Amc$ are equivalent as PA (over finite words), hence, there is also an accepting run of $\hat{\Amc}$ on $u$, say $\hat r = \hat{r}_1 \dots \hat{r}_z$. 
Note that $\hat{r}_z = (p_{z-1}, \alpha_z, \vbf_z, f)$, hence $\Amc'$ has a transition of the form $\delta = (p_{z-1}, \alpha_z, \vbf_z, q_S)$, hence we can simulate this partial run of $\Amc'$ with $\hat{r}_1 \dots \hat{r}_{z-1}\delta$.

Up to this point we have only used the first $d + |Q| - 1$ counters, and will now only use the last $d$ counters.
By definition of $C(q_S)$ it remains to show that $\beta \in P_\omega(\Amc_S)$ (note that we already had an accepting hit and have thus removed the base vectors from $C(q_S)$). Observe that we can safely remove any $\varepsilon$-transition in $r$ without malforming the run, as all $\varepsilon$-transitions are loops. 
Let $j_0 > i$ be minimal such that $r_{j_0} \in \Delta$, say $r_{j_0} = (f, \alpha_{z+1}, \vbf_{j_0}, p_{j_0})$. By construction there is a transition of the form $(q_S, \alpha_{z+1}, \0^{d+|Q|-1} \cdot \vbf_{j_0}, (p_{j_0}, 1))$ in $\Amc'$, which we use to continue our run in $\Amc'$ (\ie, we forget the $\varepsilon$-loops on $f$).

Now let $j_{m+1} > j_0$ be minimal such that there are $j_0 < j_1 < j_2 < \dots < j_m < j_{m+1}$ with $p_{j_k} = s_k$ for all $1 \leq k \leq m$, and $r_{j_{m+1}}$ is an accepting hit.
For all $j_{k-1} < \ell < j_k, 1 \leq k \leq m+1$, we consider the transition $r_\ell$. If $r_\ell \in \Emc$, we simply forget it. 
Otherwise, we replace $r_\ell$ by $r'_\ell = ((p_{\ell-1}, k), \gamma_\ell, \0 \cdot \vbf_\ell, ((p_\ell, k))$. 
Furthermore, for $1 \leq k \leq m$ we replace $r_{j_k}$ by $r_{j_k}' = ((p_{j_k - 1}, k), \gamma_{j_k},$ $\0 \cdot \vbf_{j_k}, (p_{j_k}, k+1))$, and finally $r_{j_{m+1}}$ by $r'_{j_{m+1}} = ((p_{j-1}, m+1), \0 \cdot \vbf_j, q_S)$.

Observe that for the partial run $r' = r'_{j_0} \dots r'_{j_{m+1}}$ we have $\rho(r') \in C(q_S)$, as $\rho(r_{i+1} \dots r_{j_{m+1}})$ $\in C$, since we have only removed $\varepsilon$-transitions (on states that we have all seen by construction), and $C(q_S)$ is defined accordingly.

We can now iterate the construction and obtain an accepting run of $\Amc_S$ on $\beta$. This concludes the forward direction. 

\medskip

$\Leftarrow$ We now show $P_\omega(\Amc')\subseteq P_\omega(\Amc)$. Let $\alpha \in P_\omega(\Amc')$ with accepting run $r' = r'_1 r'_2 r'_3 \dots$ with $r'_i = (p'_{i-1}, \alpha_i, \vbf'_i, p'_i)$. Let $j_1 < j_2 < \ldots$ be the positions such that $r'_{j_i}$ is an accepting hit for all $i \geq 1$. We proceed by proving a sequence of claims. 

\medskip
\noindent 
\textbf{Claim 1}: $\0 \in \vec{C} \subseteq \vec{C} - C_{s_1} - \dots - C_{s_m}$. This is immediate from the fact that $\vec C$ has no base vectors. 

\medskip
\noindent 
\textbf{Claim 2}: For all $k < j_1$ we have $p'_k \in Q$ (recall that $Q$ is the subset of the state set of~$\Amc'$ that belongs to $\hat \Amc)$. Hence $\rho(r'_1 \dots r'_{j_1}) \in \hat{C} \cdot \{\0\}$. 
The claim is immediate by the fact that~$\hat \Amc$ as a subautomaton of $\Amc'$ has no accepting states and the transition from $\hat \Amc$ to some $\Amc_S$ leads to the accepting state $q_S$, that is $r'_{j_1} = (p'_{j_1-1}, \alpha_{j_1}, \vbf'_{j_1}, q_S)$ with $p'_{j_1-1} \in Q$.

\medskip
\noindent 
\textbf{Claim 3}: $\rho(r'_{j_1+1} \dots r'_j) \in \{\0\} \cdot \Nbb^d$ for every $j \geq j_1+1$. Hence $\vbf'_j$ can be written as $\0 \cdot \vbf_j$ for some $\vbf_j \in \Nbb^d$. Furthermore, every $p'_j$ is either $q_S$ or of the form $(s_k, \ell)$ for some $k \leq m$ and $\ell \leq m+1$. This is immediate by construction of $\Amc'$, since after the first accepting hit we have switched to some $\Amc_S$.  Define $p_j = f$ if $p'_j = q_S$, and $p_j = s_k$ if $p'_j = (s_k, \ell)$. Similarly, let $r_j = (p_{j-1}, \alpha_j, \vbf_j, p_j)$.

\medskip
\noindent 
\textbf{Claim 4}: All states of $S$ are visited between every two consecutive $j_\ell, j_{\ell+1}$. This follows from the fact that in order to visit $q_S$ in $\Amc_S$ again, we have to run through all layers of~$\Amc_S$.

\medskip
\noindent 
\textbf{Claim 5}: $\rho(r'_{j_\ell+1} \dots r'_{j_{\ell+1}}) \in \{\0\} \cdot (\vec{C} - C_{s_1} - \dots - C_{s_m})$, say $\rho(r'_{j_\ell+1} \dots r'_{j_{\ell+1}}) = \vbf_{j_\ell} - \vbf_{s_1} - \dots - \vbf_{s_m}$ with $\vbf_{s_i} \in C_{s_i}$. This holds by Claim 3 and the construction of $\Amc_S$, since $\epsilon$-loops have been removed. Let $j_\ell < k_1 < k_2 < \dots < k_m < j_{\ell + 1}$ be the positions in the partial run where we change the layers, \ie, $r'_{k_i} = ((s_i, i), \alpha_{k_i}, \vbf'_i, (p_{k_i}, i+1))$. These positions exist by Claim 4.

\medskip
\noindent 
\textbf{Claim 6}: For every $\vbf_{s_i} \in C_{s_i}$ there is a sequence of $\varepsilon$-loops $\bar\lambda_i$ on $s_i$ in $\Amc$ with $\rho(\bar\lambda_i) = \vbf_{s_i}$. This is immediate by the choice of $C_{s_i}$.

\medskip
\noindent 
\textbf{Claim 7}: The run $\bar r_\ell = r_{j_\ell+1} \dots \bar\lambda_1 r_{k_1} \dots \bar\lambda_2r_{k_2} \dots \bar\lambda_m r_{k_m} \dots r_{j_{\ell + 1}}$ of $\Amc$ is equivalent to $\bar r'_\ell=r'_{j_\ell} \dots r'_{j_{\ell+1}}$ of $\Amc'$ in the sense that $\rho(\bar r_\ell) = \rho(\bar r'_\ell) - \vbf_{s_1} - \dots - \vbf_{s_m} \in \bar{C} - C_{s_1} - \dots - C_{s_m}$, hence $\rho(\bar r_\ell) \in \vec{C}$. This claim follows from Claim 5 and Claim 6.

\smallskip
We are ready to finish the proof. By \Cref{lem:KR-finite} there exists an accepting run $\hat r$ of $\alpha_1\ldots \alpha_{j_1}$ in $\hat \Amc$ (as a PA on finite words). 
By construction we have $\rho(\hat r)\in \hat C\cdot \{\0\}$. 
We construct the run $r$ of $\Amc$ as follows. In $r'$ we replace $r'_1\ldots r'_{j_1}$ by $\hat r$. 
Now, for $\ell\geq 1$ we replace $r'_{j_{\ell}+1}\ldots r'_{j_{\ell+1}}$  by $\bar r_\ell$ as constructed in Claim 7. 
Note that this is a valid run of $\Amc$. 
Furthermore, by Claim 7 we have accepting hits of $\Amc$ at positions $(|\hat r|-j_1)+j_\ell$ for all $\ell\geq 1$. 
Hence, $r$ is accepting in $\Amc$. 
\hfill $\lrcorner$

\medskip
We have proved that the MPPBA $\Amc'$ is equivalent to $\Amc$. We conclude the proof of the lemma with \Cref{lemma:GPPBA}. 
\end{proof}

By combining the previous lemmas we conclude the main theorems of this section (see also \Cref{fig:equivalencesPrefix}). 

\begin{theorem}
\label{thm:epsilon-elimination}
The class of $\epsilon$-PPBA recognizable $\omega$-languages coincides with the class of PPBA recognizable $\omega$-languages. In other words, $\epsilon$-PPBA admit $\epsilon$-elimination. 
\end{theorem}

\begin{theorem}
\label{thm:ppba-equivalence}
Let $L$ be an $\omega$-language. Then the following statements are equivalent.
\begin{enumerate}
    \item $L$ is PPBA-recognizable.
    \item $L$ is MPPBA-recognizable.
    \item $L$ is $\varepsilon$-PPBA-recognizable.
    \item $L$ is $\varepsilon$-MPPBA-recognizable.
\end{enumerate}
\end{theorem}

    \begin{figure}
    \centering
	    \begin{tikzpicture}[%
	      node distance=27mm,>=Latex,
	      initial text="", initial where=below left,
	      every state/.style={rectangle,rounded corners,draw=black,thin,fill=black!5,inner sep=1mm,minimum size=6mm},
	      every edge/.style={draw=black,thin}
	    ]
	    \node[state] (PPBA) {PPBA};
	    \node[state,above right of = PPBA] (ePPBA) {$\varepsilon$-PPBA};
        \node[state,above left  of = PPBA] (MPPBA) {MPPBA};
        \node[state,above left  of = ePPBA] (eMPPBA) {$\varepsilon$-MPPBA};
 
	    \path[->, dashed]
          (PPBA)  edge [bend left  = 10] (ePPBA)
          (PPBA)  edge [bend right = 10] (MPPBA)
          (MPPBA) edge [bend right = 10] (eMPPBA)
          (ePPBA) edge [bend left  = 10] (eMPPBA)
	    ; 

     	\path[->] 
          (eMPPBA) edge [bend left=10] node[above right] {\Cref{lemma:GPPBA}} (ePPBA)
          (ePPBA) edge [bend left=10] node[below right] {\Cref{thm:epsilon-elimination}} (PPBA)
          (MPPBA) edge [bend right=10] node[below left] {\Cref{lemma:GPPBA}} (PPBA)
	    ;

	    \end{tikzpicture}        
        \caption{Equivalences of the prefix models. Dashed lines represent obvious containment.}
        \label{fig:equivalencesPrefix}
    \end{figure}
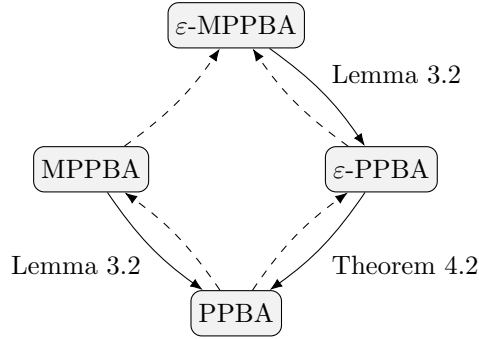

\subsection{$\varepsilon$-elimination for SPBA}

Finally, we prove that the reset models admit $\epsilon$-elimination. In this section, because we deal with SPBA, where every visit of an accepting state is resetting, we speak of resetting states instead of accepting states. 

\begin{lemma}
\label{lem:SPBAepselim}
For every $\varepsilon$-SPBA $\Amc$ there is an equivalent SPBA. In other words, SPBA admit $\varepsilon$-elimination.
\end{lemma}

\begin{proof} 
    Let $\Amc = (Q, \Sigma, q_0, \Delta, \Emc, F, C)$ be an $\varepsilon$-SPBA of dimension $d$. We assume w.l.o.g.\ that~$q_0$ has no ingoing transitions (this can be achieved by introducing a fresh copy of~$q_0$). Furthermore, we assume that $F \neq \varnothing$ (otherwise $S_\omega(\Amc) = \varnothing$). 
    Let the states of $Q$ be ordered arbitrarily, say $Q = \{q_0, \dots, q_{n-1}\}$.
    We construct an equivalent SPBA $\Amc' = (Q', \Sigma, q_0, \Delta', F', C')$ of dimension $d + n$. In the beginning, $\Amc'$ is a copy of $\Amc$ (keeping the $\epsilon$-transitions for now), which is modified step-by-step.
    The purpose of the new counters is to keep track of the states that have been visited (since the last reset). Initially, we hence modify the transitions as follows:
    for every transition $(q_i, \gamma, \vbf, q_j) \in \Delta \cup \Emc$ we replace $\vbf$ by $\vbf \cdot \ebf_j^n$. 
    
    Let $p, q \in Q$. Assume there is a sequence of transitions $\tilde\lambda = r_1 \dots r_j \dots r_k \in \Emc^* \Delta \Emc^*$; $1\leq j \leq k \leq 2n+1$, where 
    \begin{itemize}
        \item $r_j = (p_{j-1}, a, \vbf_j, p_j) \in\Delta$, and 
        \item $r_i = (p_{i-1}, \varepsilon, \vbf_i, p_i) \in \Emc$ for all $i \neq j, i \leq k$,
        \item such that $p_0 = p, p_k = q$, and $p_i\neq p_\ell$ for $i,\ell\leq j$ and $p_i\neq p_\ell$ for $i,\ell\geq j$, and
        \item all internal states are non-resetting, \ie, $p_i \notin F$ for all $0 < i < k$.
    \end{itemize}  
    
    Then we introduce the \emph{shortcut} $(p, a, \rho(\tilde\lambda), q)$, where $\rho(\tilde\lambda)$ is computed already with respect to the new counters, tracking that the $p_i$ in $\tilde\lambda$ have been visited, \ie, the counters corresponding to the $p_i$ in this sequence have non-zero values. 
    
    Let $p,q\in Q$. We call a (possibly empty) sequence $\lambda = r_1 \dots r_k \in \Emc^*; k\geq 0$ with $r_i = (p_{i-1}, \varepsilon, \vbf_i, p_i)$ and $p_0 = p, p_k = q$ a \emph{no-reset $\varepsilon$-sequence} from $p$ to~$q$ if all internal states are non-resetting, \ie, $p_i \notin F$ for all $0 < i < k$. A \emph{no-reset $\epsilon$-path} is a no-reset sequence such that $p_i\neq p_j$ for $i\neq j$. 
    Observe that the set of no-reset $\varepsilon$-paths from $p$ to $q$ is finite, as the length of each path is bounded by $n-1$.
    We call the pair $(p,q)$ a \emph{$C$-pair} if there is a no-reset $\varepsilon$-sequence $r$ from $p$ to $q$ with $\rho(r) \in C$, where $\rho(r)$ is computed in $\Amc$. 
    
    Let $S=(f_1,\ldots, f_\ell)$ be a non-empty sequence of pairwise distinct resetting states (note that this implies $\ell \leq n$). We call $S$ a \emph{$C$-sequence} if each $(f_i,f_{i+1})$ is a $C$-pair.
    
    For all $p,q\in F$ and $C$-sequences $S$ such that $p=f_1$ if $p\in F$ and $q=f_\ell$ if $q\in F$, we introduce a new state $(p,S,q)$.
    We add $(p,S,q)$ to $F'$, that is, we make the new states resetting. 
    State $(p,S,q)$ will represent a partial run of the automaton with only $\epsilon$-transitions starting in $p$, visiting the resetting states of $S$ in that order, and ending in $q$.

    Observe that in the following we introduce only finitely many transitions by the observations made above, we will not repeat this statement in each step.
    Let $p, q \in Q$ and $S = (f_1, \dots, f_\ell)$ be a $C$-sequence.
    For every transition of the form $(s, a, \vbf, p) \in \Delta$ we insert new transitions $\{(s, a, \vbf + \rho(\lambda), (p,S,q)) \mid \lambda$ is a no-reset $\varepsilon$-path from $p$ to $f_1\}$ to~$\Delta'$. 
     Similarly, for every transition of the form $(q, a, \vbf, t) \in \Delta$ we insert new transitions $\{((p,S,q), a, \vbf + \rho(\lambda), t) \mid \lambda \text{ is a no-reset $\varepsilon$-path from $f_\ell$ to $q$}\}$ to $\Delta'$. 
    Again this set is finite. 
    Additionally, let $p', q' \in Q$ and $S' = (f_1', \dots, f'_k)$ a $C$-sequence. For every sequence $\tilde\lambda=\lambda \delta \lambda'$ where $\lambda$ is a no-reset $\varepsilon$-path from $f_\ell$ to $q$, $\delta = (q, a, \vbf, p')$, and $\lambda'$ is a no-reset $\varepsilon$-path from~$p'$ to $f'_1$
    we add the shortcuts $\big((p,S,q), a, \rho(\tilde{\lambda}), (p',S',q')\big)$ to $\Delta'$. 
    
    Lastly, we connect the initial state $q_0$ in a similar way (recall that we assume that $q_0$ has no ingoing transitions, and in particular no loops). 
    For every transition $(p, a, \vbf, q) \in \Delta$ and every $C$-sequence $S = (f_1, \dots, f_l)$ with the property that $(q_0, f_1)$ is a $C$-pair and there is a no-reset $\varepsilon$-path $\lambda$ from $f_\ell$ to $p$, we introduce the transition $(q_0, a, \rho(\lambda) + \vbf, q)$ for every such path $\lambda$. 
    Additionally, for every $C$-sequence $S' = (f_1', \dots f'_k)$ such that there is a no-reset $\varepsilon$-path~$\lambda'$ from $q$ to $f_1'$, we introduce the transition $\big(q_0, a, \rho(\lambda) + \vbf + \rho(\lambda'), (q,S',t)\big)$ for all such paths $\lambda, \lambda'$ and $t\in Q$. 
    Furthermore, for every no-reset $\varepsilon$-path $\hat\lambda$ from $q_0$ to~$p$, we introduce the transition $(q_0, a, \rho(\hat{\lambda}) + \vbf + \rho(\lambda'), (q,S',t))$ for all $t \in Q$. 
    A reader who is worried that we may introduce too many transitions at this point shall recall that $(q,S',t)$ has no outgoing transition if there does not exist a no-reset $\epsilon$-path from $f_k'$ to $t$.
    Finally, we delete all $\epsilon$-transitions.
    
    We define $C'$ similar to the construction due to Klaedtke and Ruess \cite{klaedtkeruess}. For every $q \in Q$ we define $C_q = \sum \hat{\Bmc}_{q,q}$, where $\hat{\Bmc}_{q,q}$ is defined as $\Bmc_{q,q}$ but without any accepting states, that is, $\hat{\Bmc}_{q,q} = (Q \setminus F, \Gamma, q, \{(p, \vbf, p') \mid (p, \varepsilon, \vbf, p') \in \Emc, p, p' \notin F\},\{q\})$ for a suitable alphabet $\Gamma \subseteq \Nbb^d$. 
    Then $C' = \{\vbf \cdot (x_0, \dots, x_{n-1}) \mid \vbf + \ubf \in C, \ubf \in \sum_{x_i \geq 1} C_{q_i}\}$. 
    By this, we substract the $C_{q_i}$ if the counter for $q_i$ is greater or equal to one, that is, the state has been visited. 
    This finishes the construction. 
    
    We now prove that $\Amc'$ is equivalent to $\Amc$. In the one direction we compress the run by using the appropriate shortcuts, in the other direction we unravel it accordingly. 

    \medskip
    $\Rightarrow$ To show that $S_\omega(\Amc)\subseteq S_\omega(\Amc')$, let $\alpha \in S_\omega(\Amc)$ with accepting run $r = r_1 r_2 r_3 \dots$. If there are no $\varepsilon$-transitions in $r$, we are done (as $r$ is also an accepting run of $\Amc'$ on $\alpha$).
    
    Otherwise, we construct an accepting run $r'$ of $\Amc'$ on $\alpha$ by replacing maximal \mbox{$\varepsilon$-sequences} in $r$ step-by-step. Let $i$ be minimal such that $r_i \dots r_j$ is a maximal $\varepsilon$-sequence. 
    Let $r_i = (p_{i-1}, \varepsilon, \vbf_i, p_i), r_j = (p_{j-1}, \varepsilon, \vbf_j, p_j)$, and $r_{j+1} = (p_j, \alpha_z, \vbf_{j+1}, p_{j+1})$. It might be the case that $i = 1$, \ie, the run $r$ starts with an $\varepsilon$-transition leaving $q_0$. Otherwise $i > 1$ and we can write $r_{i-1} = (p_{i-2}, \alpha_{z-1}, \vbf_{i-1}, p_{i-1})$.
    By allowing the empty sequence, we may assume that there is always a second (possibly empty) maximal $\varepsilon$-sequence $r_{j+2} \dots r_{k}$ starting directly after $r_{j+1}$. 
    We distinguish (the combination of) the following cases.
    \begin{itemize}
        \item At least one state in $r_i \dots r_j$ is resetting, \ie, there is a position $i-1 \leq \ell \leq j$ such that $p_\ell \in F$ (F) or not (N).
        \item At least one state in $r_{j+2} \dots r_k$ is resetting, \ie, there is a position $j+1 \leq \ell' \leq k$ such that $p_{\ell'} \in F$ (F) or not (N). If $r_{j+2} \dots r_k$ is empty, we are in the case~(N).
    \end{itemize}

Hence, we consider four cases in total.
    
    \medskip
    $\bullet$ Case (NN). That is, there is no resetting state in $r_i \dots r_k$. Note that the $\varepsilon$-sequence $r_i \dots r_j$ can be decomposed into an $\varepsilon$-path and $\varepsilon$-cycles as follows. 
    If we have $p_{i_1} \neq p_{j_1}$ for all $i \leq i_1 < j_1 \leq j$ we are done as $r_i \dots r_j$ is already an $\varepsilon$-path. Otherwise let $i_1 \geq i$ be minimal such that there is $j_1 > i_1$ with $p_{i_1} = p_{j_1}$, that is, $r_{i_1+1} \dots r_{j_1}$ is an $\varepsilon$-cycle. 
    If $r_i \dots r_{i_1} r_{j_1+1} \dots r_j$ is an $\varepsilon$-path, we are done. 
    Otherwise, let $i_2 > j_1$ be minimal such that there is $j_2 > i_2$ with $p_{i_2} = p_{j_2}$, that is, $r_{i_2+1} \dots r_{j_2}$ is an $\varepsilon$-cycle. 
    Then again, if $r_i \dots r_{i_1} r_{j_1+1} \dots r_{i_2} r_{j_2+1} \dots r_j$ is an $\varepsilon$-path, we are done. 
    Otherwise, we can iterate this argument and obtain a set of $\varepsilon$-cycles $r_{i_1+1} \dots r_{j_1}, \dots, r_{i_m+1} \dots r_{j_m}$ for some $m$, and an $\varepsilon$-path $\hat{r}_{i,j} = r_i \dots r_{i_1} r_{j_1+1} \dots r_{i_{m}} r_{j_m+1} \dots r_j$ which partition $r_i \dots r_j$. 
    Now observe that $\rho(r_{i_1+1} \dots r_{j_1}) + \dots + \rho(r_{i_m+1} \dots r_{j_m}) \in C_{p_{i_1}} + \dots + C_{p_{i_m}}$. 
    We can do the same decomposition for the $\varepsilon$-sequence $r_{j+2} \dots r_k$ into a set of $\varepsilon$-cycles and an $\varepsilon$-path $\hat{r}_{j+2,k}$. 
    By the construction of $\Delta'$, there is a shortcut $\delta = (p_{i-1}, \alpha_z, (\rho(\hat{r}_{i,j}) + \vbf_{j+1} + \rho(\hat{r}_{j+2,k})) \cdot \hat\ebf^n, p_k)$, where $\hat\ebf^n$ is the $n$-dimensional vector counting the states appearing in $\hat r_{i,j}$ and $\hat r_{j+2,k}$ and the state $p_{j+1}$. 
    By the construction of $\Delta'$ and $C'$, we may subtract all $\varepsilon$-cycles that have been visited in $r_i \dots r_k$, hence, we may replace $r_i \dots r_k$ by $\delta$ to simulate exactly the behavior of~$\Amc$.

    \medskip
    $\bullet$
    Case (NF). That is, there is no resetting state in $r_i \dots r_j$ but at least one resetting state in $r_{i+2} \dots r_k$ (in particular, this sequence is not empty). 
    Let $\ell_1, \dots, \ell_m$ denote the positions of resetting states in $r_{i+2} \dots r_k$, and let $\ell_0 < \ell_1$ be maximal such that~$\ell_0$ is resetting (this is before $r_i$, and if such an $\ell_0$ does not exist, let $\ell_0 = 0$), \ie, $\ell_0$ is the position of the last reset before the reset at position $\ell_1$. 
    As $r$ is an accepting run, the sequence $S = (\ell_1, \dots, \ell_m)$ is a $C$-sequence (we may assume that all states in~$S$ are pairwise distinct, otherwise there is a reset-cycle, which can be ignored). 
    In the same way as in the previous case we can partition the $\varepsilon$-sequence $r_i \dots r_j$ into an $\varepsilon$-path $\hat{r}_{i,j}$ and a set of $\varepsilon$-cycles, which may be subtracted from $C$. 
    Likewise, we can partition the sequence $r_{j+2} \dots r_{\ell_1}$ into an $\varepsilon$-path $\hat{r}_{j+2,\ell_1}$ and $\varepsilon$-cycles with the same property. 
    By the construction of $\Delta'$ there is a shortcut $(p_{i-1}, a, \rho(\hat{r}_{i,j}) + \vbf_{j+1}, p_{j+1})$ and hence a transition $\delta = (p_{i-1}, a, \rho(\hat{r}_{i,j}) + \vbf_{j+1} + \rho(\hat{r}_{j+2, \ell_1}), (p_{j+1}, S,p_k))$ (note that this is also the case if~$i = 1$). 
    Thus, we replace $r_i \dots r_k$ by~$\delta$. 
    In particular, $\rho(r_{\ell_0+1} \dots r_{i-1}\delta)$ can be obtained from $\rho(r_{\ell_0+1} \dots r_{\ell_1})$ by subtracting all $\varepsilon$-cycles that have been visited within this partial run. 
    Furthermore, observe that $\rho(r_{\ell_1+1} \dots r_{\ell_2}) \in C, \dots, \rho(r_{\ell_{m-1}+1} \dots r_{\ell_m}) \in C$ depend only on the automaton, and not the input word. 
    As the counters are reset in $r_{\ell_m}$, we may continue the run from $\delta$ the same way as in $r_k$, using an appropriate transition from~$\Delta'$ that adds the vector $\rho(\hat{r}_{\ell_{m}+1, k})$, thus respecting the acceptance condition.

    \medskip
    $\bullet$
    Case (FN). Similar to (NF), but this time we replace $r_{i-1}r_i \dots r_j$ by an appropriate transition into a state of the form $(p_{i-2}, \alpha_z, \vbf, (p_{i-1}, S, p_j))$ for a suitable $C$-sequence $S$ and vector $\vbf$, followed by a shortcut leading to $p_k$. If $i = 0$ (we enter a $C$-sequence before reading the first symbol), we make use of the transitions introduced especially for $q_0$.

    \medskip
    $\bullet$
    Case (FF). Similar to (FN) and (NF), but we transition from a state of the form $(p_{i-1}, S, p_j)$ into a state of the form $(p_{j+1}, S', p_k)$ for suitable $C$-sequences $S, S'$, again respecting the case $i = 0$.
    \hfill $\lrcorner$

    \medskip
    $\Leftarrow$ To show that $S_\omega(\Amc')\subseteq S_\omega(\Amc)$ we unravel the shortcuts and $(p, S,q)$-states introduced in the construction. Let $\alpha \in S_\omega(\Amc')$ with accepting run $r' = r'_1 r'_2 r'_3 \dots$. We replace every transition $r'_i \in \Delta' \setminus \Delta$ (\ie, transitions that do not appear in $\Amc$) by an appropriate sequence of transitions in $\Amc$.
    Let $i \geq 1$ be minimal such that $r'_i$ is a transition in $\Delta' \setminus \Delta$.

    We distinguish the form of $r'_i$ and show that the possible forms correspond one-to-one to the cases in the forward direction.

      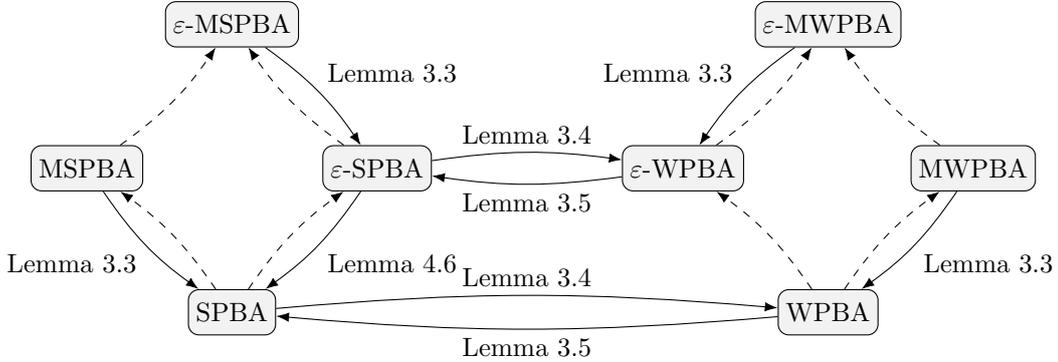
\begin{figure}
    \centering
	    \begin{tikzpicture}[%
	      node distance=27mm,>=Latex,
	      initial text="", initial where=below left,
	      every state/.style={rectangle,rounded corners,draw=black,thin,fill=black!5,inner sep=1mm,minimum size=6mm},
	      every edge/.style={draw=black,thin}
	    ]
	    \node[state] (SPBA) {SPBA};
	    \node[state,above right of = SPBA] (eSPBA) {$\varepsilon$-SPBA};
        \node[state,above left  of = SPBA] (MSPBA) {MSPBA};
        \node[state,above left  of = eSPBA] (eMSPBA) {$\varepsilon$-MSPBA};

        \node[state, right = 25mm of eSPBA] (eWPBA) {$\varepsilon$-WPBA};
	    \node[state,below right of = eWPBA] (WPBA) {WPBA};
        \node[state,above right of = WPBA]  (MWPBA) {MWPBA};
        \node[state,above right of = eWPBA] (eMWPBA) {$\varepsilon$-MWPBA};
 
	    \path[->, dashed]
          (SPBA)  edge [bend left  = 10] (eSPBA)
          (SPBA)  edge [bend right = 10] (MSPBA)
          (MSPBA) edge [bend right = 10] (eMSPBA)
          (eSPBA) edge [bend left  = 10] (eMSPBA)

          (WPBA)  edge [bend right = 10] (eWPBA)
          (WPBA)  edge [bend left  = 10] (MWPBA)
          (MWPBA) edge [bend left  = 10] (eMWPBA)
          (eWPBA) edge [bend right = 10] (eMWPBA)
	    ; 

     	\path[->] 
          (eMSPBA) edge [bend left=10] node[above right] {\Cref{lem:MSPBAtoSPBA}} (eSPBA)
          (eSPBA) edge [bend left=10] node[below right] {\Cref{lem:SPBAepselim}} (SPBA)
          (MSPBA) edge [bend right=10] node[below left] {\Cref{lem:MSPBAtoSPBA}} (SPBA)

          (eSPBA) edge [bend left=8] node[above] {\Cref{lem:SPBAtoWPBA}} (eWPBA)
          (eWPBA) edge [bend left=8] node[below] {\Cref{lem:WPBAtoSPBA}} (eSPBA)
          (SPBA)  edge [bend left=5] node[above] {\Cref{lem:SPBAtoWPBA}} (WPBA)
          (WPBA)  edge [bend left=5] node[below] {\Cref{lem:WPBAtoSPBA}} (SPBA)

          (eMWPBA) edge [bend right=10] node[above left] {\Cref{lem:MSPBAtoSPBA}} (eWPBA)
          (MWPBA) edge [bend left=10] node[below right] {\Cref{lem:MSPBAtoSPBA}} (WPBA)
	    ;

	    \end{tikzpicture}        
        \caption{Equivalences of the reset models. Dashed lines represent obvious containment.}
        \label{fig:equivalencesReset}
    \end{figure}
    
    \begin{itemize}
        \item Case (NN). The case that $r'_i = (p, a, \rho(\tilde{\lambda}), q)$ is a shortcut, \ie, $\tilde{\lambda} \in \Emc^* \Delta \Emc^*$, corresponds to the case (NN). 
        In particular, there are no accepting states in $r$. Let $k < i$ be the position of the last reset before $r'_i$, and $k'$ the position of the first reset after $r'_i$, where $k' = i$ if $r'_i$ transitions into a resetting state. By the acceptance condition we have $\rho(r'_{k+1} \dots r'_{k'}) \in C - (\sum_{q \in Q'} C_q)$ for some set $Q' \subseteq Q$ based on the counter values. Hence, we can replace $r'_i$ by the partial run $\tilde{\lambda}$ filled with possible $\varepsilon$-cycles on some states in $Q'$.
        \item Case (NF). The case that $r'_i = (s, a, \vbf + \rho(\lambda), (p,S,q))$ such that $S = (f_1, \dots f_\ell)$ is a $C$-sequence, there is a transition $\delta = (s, a, \vbf, p) \in \Delta$ and $\lambda$ is a no-reset $\varepsilon$-path from $p$ to $f_1$, corresponds to the case (NF). 
        By the definition of $C$-sequence there is a sequence $r_{f_1, f_\ell}$ of $\varepsilon$-transitions in $\Amc$ starting in $f_1$, ending in $f_\ell$, visiting the resetting states $f_1$ to $f_\ell$ (in that order) such that the reset-acceptance condition is satisfied on every visit of one the accepting states. 
        Then we can replace $r'_i$ by $\delta \lambda r_{f_1, f_\ell}$, possibly again filled with some $\varepsilon$-cycles based on the state counters of~$\lambda$, similar to the previous case. Note that at this point we do not yet unravel the path from $f_\ell$ to $q$, as it depends on how the run $r'$ continues (as handled by the next two cases).
        \item Case (FN). The case that $r'_i = ((p,S,q), a, \vbf + \rho(\lambda), t)$ such that $S = (f_1, \dots f_\ell)$ is a $C$-sequence, there is a transition $\delta = (q, a, \vbf, t) \in \Delta$ and $\lambda$ is a no-reset $\varepsilon$-path from~$f_\ell$ to $q$, corresponds to the case (FN). 
        Similar to the previous case, we can replace~$r'_i$ by~$\lambda\delta$, possibly again amended with some $\varepsilon$-cycles based on the state counters of~$\lambda$.
        If $i = 1$, the transition might also be of the form $r'_1 = (q_0, \alpha_1, \rho(\lambda) + \vbf, t)$ such that $S$ is a $C$-sequence with the property that $(q_0, f_1)$ is a $C$-pair. Then there is a sequence of $\varepsilon$-transitions $r_{q_0, f_\ell}$ in $\Amc$ as above. 
        Then we replace $r'_1$ by $r_{q_0, f_\ell} \lambda \delta$ (with possible $\varepsilon$-cycles) instead.
        \item Case (FF). The case that $r'_i = ((p,S,q), a, \rho(\tilde{\lambda}), (p', S', q')$ such that $S = (f_1, \dots f_\ell)$ and $S' = (f'_1, \dots, f'_k)$ are $C$-sequences, there is a transition $\delta = (q, a, \vbf, p') \in \Delta$ and $\tilde{\lambda} = \lambda \delta \lambda'$,  where $\lambda$ is a no-reset $\varepsilon$-path from $f_\ell$ to $q$ and $\lambda'$ is a no-reset $\varepsilon$-path from $p'$ to $f_1'$, corresponds to the case (FF).
        This case is basically the union of the previous cases. There is a sequence $r_{f'_1, f'_k}$ of $\varepsilon$-transitions in $\Amc$, as in the case (RF). Hence, we replace $r'_i$ by $\tilde{\lambda} r_{f'_1, f'_k}$ (with possible $\varepsilon$-cycles).
        If $i = 1$, the transition might also be of the form $r'_1 = (q_0, \alpha_1, \rho(\lambda) + \vbf + \rho(\lambda'), (p', S', q'))$ such that $(q_0, f_1)$ is a $C$-pair. Then there is a sequence of $\varepsilon$-transitions $r_{q_0, f_\ell}$ in $\Amc$ as above, and we replace $r'_1$ by $r_{q_0, f_\ell} \tilde{\lambda} r_{f'_1, f'_k}$ (with possible $\varepsilon$-cycles). \hfill $\lrcorner$
    \end{itemize}

This finishes the proof of the lemma. 
\end{proof}

Concluding this section, we have proved the equivalence of all reset models, as stated in the following theorem. See \Cref{fig:equivalencesReset} for an illustration.

\begin{theorem}
\label{thm:reset-equivalence}
Let $L$ be an $\omega$-language. Then the following statements are equivalent.
\begin{enumerate}
    \item $L$ is SPBA-recognizable.
    \item $L$ is MSPBA-recognizable.
    \item $L$ is $\varepsilon$-SPBA-recognizable.
    \item $L$ is $\varepsilon$-MSPBA-recognizable.
    \item $L$ is WPBA-recognizable.
    \item $L$ is MWPBA-recognizable.
    \item $L$ is $\varepsilon$-WPBA-recognizable.
    \item $L$ is $\varepsilon$-MWPBA-recognizable.
\end{enumerate}
\end{theorem}

Recall that $\LReset$ denotes the class of all SPBA-recognizable $\omega$-languages, which by the above theorem can equivalently be defined using any reset model. 

%% file: k-counter.tex
\section{Equivalence of PPBA and blind counter machines and their closure \mbox{properties}}
\label{sec:blind-counter}
In this section we prove that PPBA describe the same class of $\omega$-languages as (synchronous) blind counter machines introduced by Fernau and Stiebe \cite{blindcounter}. A \emph{blind $k$-counter machine} (CM) is quintuple $\Mmc = (Q, \Sigma, q_0, \Delta, F)$ where $Q$, $\Sigma$, $q_0$ and $F$ are defined as for NFA,
and $\Delta \subseteq Q \times (\Sigma \cup \{\varepsilon\}) \times \Zbb^d \times Q$ is the set of \emph{integer labeled transitions}. In particular, the transitions of $\Delta$ are labeled with possibly negative integer vectors. Furthermore, $\varepsilon$-transitions are allowed.

A \emph{configuration} for an infinite word $\alpha = \alpha_1\alpha_2\alpha_3\dots$ of $\Mmc$ is a tuple of the form $c = (p, \alpha_1 \dots \alpha_i, \alpha_{i+1} \alpha_{i+2} \dots, \vbf) \in Q \times \Sigma^* \times \Sigma^\omega \times \Zbb^k$ for some $i \geq 0$. 
A configuration $c$ \emph{derives} into a configuration $c'$, written $c \vdash c'$, if either $c' = (q, \alpha_1 \dots \alpha_{i+1}, \alpha_{i+2} \dots, \vbf + \ubf)$ and $(p, \alpha_{i+1}, \ubf, q) \in \Delta$, or $c' = (q, \alpha_1 \dots \alpha_i, \alpha_{i+1}\alpha_{i+2} \dots, \vbf + \ubf)$ and $(p,\varepsilon, \ubf,q) \in \Delta$. 
$\Mmc$~\emph{accepts} an infinite word $\alpha$ if there is an infinite sequence of configuration derivations $c_1 \vdash c_2 \vdash c_3 \vdash \dots$ with $c_1 = (q_0, \varepsilon, \alpha, \0)$ such that for infinitely many $i$ we have $c_i = (p_i, \alpha_1 \dots \alpha_j, \alpha_{j+1} \alpha_{j+2} \dots, \0)$ with $p_i \in F$ and for all $j \geq 1$ there is a configuration of the form $(p, \alpha_1 \dots \alpha_j, \alpha_{j+1} \alpha_{j+1} \dots, q)$ for some $p, q \in Q$ in the sequence.
That is, a word is accepted if we infinitely often visit an accepting state when the counters are $\0$, and every symbol of $\alpha$ is read at some point.
We define the $\omega$-language recognized by $\Mmc$ as $L_\omega(\Mmc) = \{\alpha \in \Sigma^\omega \mid \Mmc \text{ accepts } \alpha\}$.


We show that we can effectively convert every CM into an equivalent $\varepsilon$-PPBA, which is equivalent to a PPBA by \Cref{thm:epsilon-elimination}. 
For the other direction we will make use of $\varepsilon$-transitions, \ie, our results do not yield an $\varepsilon$-elimination scheme for CM. To the best of our knowledge, it is unknown if CM without $\varepsilon$-transitions are as powerful as CM with $\varepsilon$-transitions.

\begin{lemma}
\label{lem:kcountertoPPBA} 
For every CM $\Mmc$ there is an equivalent $\varepsilon$-PPBA $\Amc$.
\end{lemma}
\begin{proof}
Let $\Mmc = (Q, \Sigma, q_0, \Delta, F)$ be a $k$-counter machine. For a vector \mbox{$(x_1, \dots, x_k) \in \Zbb^k$} we define the vector $\vbf^\pm = (x_1^+, \dots x_k^+, x_1^-, \dots x_k^-) \in \Nbb^{2k}$ as follows: if $x_i$ is positive, then $x_i^+ = x_i$ and $x_i^- = 0$. Otherwise, $x_i^+ = 0$ and $x_i^- = |x_i|$. We construct an equivalent $\varepsilon$-PPBA $\Amc = (Q, \Sigma, q_0, \Delta', \Emc', F, C)$ of dimension $2k$, where $\Delta' = \{(p, a, \vbf^\pm, q) \mid (p, a, \vbf, q) \in \Delta\}$ and $\Emc' = \{(p, \varepsilon, \vbf^\pm, q) \mid (p, \varepsilon, \vbf, q) \in \Delta\}$. Finally, we choose $C = \{(x_1, \dots, x_k, x_1, \dots, x_k) \mid x_i \in \Nbb\}$. It is now easily verified that $L_\omega(\Mmc) = P_\omega(\Amc)$.
\end{proof}

Next, we show that we can convert every PPBA with $d$ counters into an equivalent $d$-counter machine by introducing one new state.
\begin{lemma}
\label{lem:PPBAtokcounter}
For every PPBA $\Amc$ there is an equivalent CM $\Mmc$.
\end{lemma}
\begin{proof}
Let $\Amc = (Q, \Sigma, q_0, \Delta, F, C)$ be a PPBA of dimension $d$. We can assume that~$C$ is linear by \Cref{lemma:indep} and since CM are closed under union \cite{blindcounter}.
We construct a $d$-counter machine $\Mmc$ that simulates $\Amc$ as follows: $\Mmc$ consists of a copy of $\Amc$ where the accepting states have additional $\varepsilon$-transitions labeled with the negated period vectors of $C$
We only need to consider the base vector of $C$ a single time, hence we introduce a fresh initial state $q_0'$ and a $\varepsilon$-transition from $q'_0$ to $q_0$ labeled with the negated base vector of~$C$.
Observe that a vector $\vbf$ lies in $C = \{b_0 + b_1z_1 + \dots + b_\ell z_\ell \mid z_1, \dots, z_\ell\}$ if and only if $\vbf - b_1z_1 - \dots - b_\ell z_\ell - b_0 = \0$ for some $z_i$. 
Intuitively, $\Mmc$ computes the vector $\vbf$ in the copies of $Q$ and guesses the $z_i$ in the accepting states.
We construct $\Mmc = (Q \cup \{q_0'\}, \Sigma, q_0', \Delta', F)$ where
\[
\Delta' = \Delta \cup \{(q_0', \varepsilon, -b_0, q_0\} \cup \{(q_f, \varepsilon, -b_i, q_f) \mid q_f \in F, i \leq \ell\}.
\]

It is now easily verified that $P_\omega(\Amc) = L_\omega(\Mmc)$.
\end{proof}

From the previous two lemmas follows the equivalence of PPBA and CM.
\begin{corollary}
\label{cor:equivalencePPBAKM}
The classes of PPBA-recognizable $\omega$-languages and CM-recognizable $\omega$-languages coincide.
\end{corollary}
As shown by Fernau and Stiebe, the class of CM-recognizable $\omega$-languages are closed under union, and closed under intersection with $\omega$-regular languages, but not closed under intersection and complement. Hence, the class of PPBA-recognizable $\omega$-languages has the same closure properties. 
At this point we briefly mention that the class of SPBA-recognizable $\omega$-languages is also closed under union and intersection with $\omega$-regular languages, and not closed under intersection and complement (which can be shown using very similar arguments).
\begin{observation}
\label{cor:closure}
The classes of PPBA-recognizable $\omega$-languages and SPBA-recognizable \mbox{$\omega$-languages} are closed under union, and closed under intersection with $\omega$-regular languages, but not under intersection and complement.
\end{observation}

We conclude this section by showing that the classes of PPBA-recognizable and SPBA-recognizable $\omega$-languages are also closed under left-concatenation with Parikh-recognizable languages.
We prove the following auxiliary lemma, which simplifies the proof of this statement, as well as some proofs in the following sections.
\begin{lemma}
\label{lem:normalized}
Let $\Amc = (Q, \Sigma, q_0, \Delta, F, C)$ be a PA (on finite words) with $L(\Amc) = L$. Then we can construct in polynomial time a PA $\Amc_N$ with the following properties.
\begin{itemize}
\item $L(\Amc_N) = L \setminus \{\varepsilon\}$.
\item $\Amc_N$ has a single accepting state $f$, and $f$ has no outgoing transitions.
\end{itemize}
We say that $\Amc_N$ is \emph{normalized}.
\end{lemma}
\begin{proof}
Let $\Amc_N = (Q \cup \{f\}, \Sigma, q_0, \Delta_N, \{f\}, C)$, where $\Delta_N = \Delta \cup \{(p, a, \vbf, f) \mid (p, a, \vbf, q) \in \Delta$, $q \in F\}$, that is, $\Amc_N$ guesses the position of the last symbol of the input word and moves to the new accepting state (which has no outgoing transitions) if the last symbol is read. As we need to read at least one symbol to reach $f$, we can never accept the empty word.
It is obvious that $\Amc_N$ as constructed above satisfies the conditions of being normalized and can be computed in polynomial time.
\end{proof}

\begin{lemma}
\label{lem:concatenation}
Let $L_1\subseteq \Sigma^*$ be Parikh-recognizable and $L_2\subseteq \Sigma^\omega$ be SPBA-recognizable (PPBA-recognizable). Then $L_1 L_2$ is SPBA-recognizable (PPBA-recognizable).
\end{lemma}
\begin{proof}
We start with the argument for SPBA.
Let $\Amc_1 = (Q_1, \Sigma, q_1, \Delta_1, F_1, C_1)$ be a PA of dimension $d_1$ with $L(\Amc_1) = L_1$ and $\Amc_2 = (Q_2, \Sigma, q_2, \Delta_2, F_2, C_2)$ be an SPBA of dimension~$d_2$ with $S_\omega(\Amc_2) = L_2$.
W.l.o.g.\ we assume that $\varepsilon \notin L$ and hence assume that~$\Amc_1$ is normalized, in particular that $F_1 = \{f\}$ (if $\varepsilon \in L$, we use $L_1 L_2 = (L_1 \setminus \{\varepsilon\})L_2 \cup L_2$ and the closure under union).

We choose the SPBA $\Amc = (Q_1 \cup Q_2, \Sigma, q_1, \Delta, F_1 \cup F_2, C)$, with
\begin{align*}
\Delta  =&\ \{(p, a, \vbf \cdot \0^{d_2}, q) \mid (p, a, \vbf, q) \in \Delta_1\} \\
\cup&\ \{(p, a, \0^{d_1} \cdot \vbf, q) \mid (p, a, \vbf, q) \in \Delta_2\} \\
\cup&\ \{(f, a, \0^{d_1} \cdot \vbf, q) \mid (q_2, a, \vbf, q) \in \Delta_2\}
\end{align*}
and $C = C_1 \cdot \{\0^{d_2}\} \cup \{\0^{d_1}\} \cdot C_2$.
It is straightforward to prove that
$S_\omega(\Amc) = L(\Amc_1) \cdot S_\omega(\Amc_2)$,
we just give a proof sketch.
The SPBA $\Amc$ is constructed in such a way that it starts in a copy of $\Amc_1$
and can transition from $f$ to the initial state of the copy of $\Amc_2$.
Since $\Amc_1$ is normalized it has only a single accepting state $f$ with no outgoing transition into~$\Amc_1$. In particular, in order to accept an infinite word, the automaton~$\Amc$ must transition from the copy of~$\Amc_1$ to the copy of~$\Amc_2$. 
Since $f$ is the only accepting state of $\Amc_1$ and the only state with transitions to the states of the copy of $\Amc_2$, we must have read a word from~$L(\Amc_1)$ upon reaching $f$ by the choice of $C$. After the reset in $f$, it now accepts only if the rest of the infinite word belongs to $S_\omega(\Amc_2)$ (note that the first $d_1$ counters are reset in $f$, and only the last $d_2$ counters are used after the reset). 

The proof for PPBA is similar. Instead of an SPBA, we start with a PPBA
$\Amc_2$ for~$L_2$ and construct a PPBA $\Amc$ for $L_1 L_2$.
The only difference in the construction
is to set $C = C_1 \cdot C_2$, as the counters are not reset when the automaton
visits $f$.
\end{proof}

%% file: expressiveness.tex
\section{Expressiveness of PPBA and SPBA}
\label{sec:characterization}

In this section, we show that the class of PPBA-recognizable $\omega$-languages is a strict subset of the class of SPBA-recognizable $\omega$-languages.
It will be convenient to consider the class $\LPAomega = \{U_1V_1^\omega \cup \dots \cup U_nV_n^\omega \mid n \geq 1, U_i, V_i \in \LPA\}$, which is inspired by Büchi's theorem. 
Observe that this class is equivalent to the class $\Kmc_*$ mentioned by Fernau and Stiebe who have shown that the class of CM-recognizable $\omega$-languages is a strict subset of $\Kmc_*$ \cite[Lemmas 3.2 and 3.3]{blindcounter}. By \Cref{cor:equivalencePPBAKM} we obtain that the class of PPBA-recognizable $\omega$-languages is a strict subset of $\LPAomega$.

We show that $\LPAomega$ is a strict subset of the class of SPBA-recognizable languages. We begin by showing that the $\omega$-closure of Parikh-recognizable languages is SPBA-recognizable.
\begin{lemma}
\label{lemma:LomegaSPBA}
    Let $L \subseteq \Sigma^*$ be Parikh-recognizable. Then $L^\omega$ is SPBA-recognizable.
\end{lemma}
\begin{proof}
Let $\Amc = (Q, \Sigma, q_0, \Delta, F, C)$ be a PA with $L(\Amc) = L$. We show that we can construct an SPBA $\Amc'$ that simulates accepting runs of $\Amc$ infinitely often, implying $L_\omega(\Amc') = L^\omega$. 
As we have $L^\omega = (L \setminus \{\varepsilon\})^\omega$ for all languages $L$ by definition, we assume that $\Amc$ is normalized by \Cref{lem:normalized}.

We choose $\Amc' = (Q, \Sigma, q_0, \Delta', F, C)$ where $\Delta' = \Delta \cup \{(f, a, \vbf, q) \mid (q_0, a, \vbf, q) \in \Delta\}$.
We show that $S_\omega(\Amc') = L(\Amc)^\omega$.

\smallskip
$\Rightarrow$ We first show that $S_\omega(\Amc')\subseteq L(\Amc)^\omega$. Let $\alpha \in S_\omega(\Amc')$ with accepting run \mbox{$r = r_1r_2r_3\dots$}, where $r_i = (p_{i-1}, \alpha_i, \vbf_i, p_i)$. Let $0 = k_0< k_1< k_2< \ldots$ 
denote the reset positions in $r$. 
We show that we can modify the partial run $r_{j-1, j} = r_{k_{j-1}+1} r_{k_{j-1}+2} \dots r_{k_j}$ on $w_j = \alpha[k_{j-1} + 1, k_j]$ for all $j \geq 1$ such that it becomes an accepting run of~$\Amc$ on~$w_j$. 
Note that $r_{0,1}$ is already an accepting run of~$\Amc$ on $w_1$. 
Observe that for $j > 1$ the transition $r_{k_{j-1}+1}$ is a new transition of the form $(f, \alpha_{k_{j-1}+1}, \vbf_{k_{j-1}+1}, p_{k_{j-1}+1})$, as~$f$ has no outgoing transitions in $\Amc$. 
In particular, there is a transition of the form $\delta = (q_0, \alpha_{k_{j-1}+1}, \vbf_{k_{j-1}+1}, p_{k_{j-1}+1}) \in \Delta$. Hence, $\delta r_{k_{j-1}+2} \dots r_{k_{j}}$ is an accepting run of $\Amc$ on $w_j$ for all $j > 1$. Hence we have $\alpha \in L(\Amc)^\omega$. \hfill $\lrcorner$

\smallskip
$\Leftarrow$ To show that $L(\Amc)^\omega\subseteq S_\omega(\Amc')$, let $w_1 w_2 w_3 \dots \in L(\Amc)^\omega$, where $w_j \in L(\Amc)$ for all $j \geq 1$. 
Let $n_j = |w_1| + \dots + |w_j|$ denote the length of $w_1 \cdots w_j$, and let $r_{j-1, j} = r_{n_{j-1} + 1} r_{n_{j-1} + 2} \dots r_{n_j}$, where $r_i = (p_{i-1}, \alpha_{i}, \vbf_i, p_i)$ be an accepting run of $\Amc$ on $w_j$. 
Note that $p_{n_j}$ is the only accepting state in $r_{j-1,j}$. Furthermore, there is a transition $\delta_{j} = (p_{n_j}, \alpha_{n_j + 1}, \vbf_{n_j + 1}, p_{n_j + 1}) \in \Delta'$. 
Let $r'_{j-1, j} = \delta_{j-1} r_{n_{j-1} + 2} \dots r_{n_j}$ for all $j \geq 2$. Then $r = r_{0,1}r_{1,2}'r_{2,3}'\dots$ is a run of $\Amc'$ on $w_1w_2w_3\dots$ and we have a reset on each $r_{n_j}$. 
As $\rho(w_j) \in C$ for all $j \geq 1$, the run $r$ is accepting, and hence, $w_1 w_2 w_3 \dots \in S_\omega(\Amc')$. 
\end{proof}

\begin{remark}
As the $\omega$-language $\{a^n b^n \mid n \geq 1\}^\omega$ is not PPBA-recognizable (consequence of \Cref{cor:equivalencePPBAKM} and \cite[Lemma 3.3]{blindcounter}), this lemma is not true for PPBA-recognizable languages.
\end{remark}

At this point we have all ingredients to show that the class $\LPAomega$ is a subset of the class of SPBA-recognizable $\omega$-languages. We show that this inclusion is strict.

\begin{theorem}
\label{thm:LPAomegaLReset}
$\LPAomega \subsetneq \LReset$.
\end{theorem}
\begin{proof}
The inclusion follows directly from \Cref{cor:closure}, \Cref{lem:concatenation} and \Cref{lemma:LomegaSPBA}.    

We show that the inclusion is strict.
Let $L = \{a^n b^n \mid n \geq 1\}^* \cdot \{a\}^\omega \cup \{a^n b^n \mid n \geq 1\}^\omega$. We show $L \in \LReset$ but $L \notin \LPAomega$.
That $L \in \LReset$ is witnessed by the SPBA in \Cref{fig:resetpbacounterexample} with $C = \{(z,z) \mid z \in \Nbb\}$.
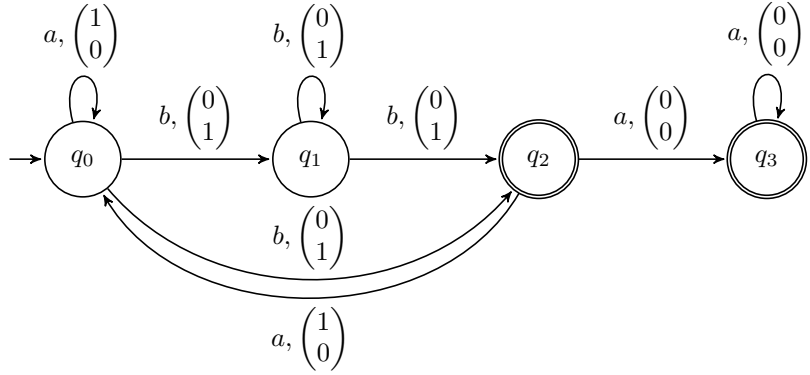
\begin{figure}
	\centering
	\begin{tikzpicture}[->,>=stealth',shorten >=1pt,auto,node distance=3cm, semithick]
	\tikzstyle{every state}=[minimum size=1.0cm]
	
	\node[initial, initial text = {}, state] (q0) {$q_0$};
	\node[state] (q1) [right of=q0] {$q_1$};
	\node[state, accepting] (q2) [right of=q1] {$q_2$};	
	\node[state, accepting] (q3) [right of=q2] {$q_3$};	

	\path
	(q0) edge [loop above] node {$a, \begin{pmatrix}1\\0\end{pmatrix}$} (q0)
	(q0) edge              node {$b, \begin{pmatrix}0\\1\end{pmatrix}$} (q1)
	(q0) edge [bend right = 50] node {$b, \begin{pmatrix}0\\1\end{pmatrix}$} (q2)
	(q1) edge [loop above] node {$b, \begin{pmatrix}0\\1\end{pmatrix}$} (q1)
	(q1) edge              node {$b, \begin{pmatrix}0\\1\end{pmatrix}$} (q2)
	(q2) edge [bend left = 60]  node {$a, \begin{pmatrix}1\\0\end{pmatrix}$} (q0)
	(q2) edge              node {$a, \begin{pmatrix}0\\0\end{pmatrix}$} (q3)	
	(q3) edge[loop above]  node {$a, \begin{pmatrix}0\\0\end{pmatrix}$} (q3)		
	;
	\end{tikzpicture}
	\caption{The SPBA for $L = \{a^n b^n \mid n \geq 1\}^* \cdot \{a\}^\omega \cup \{a^n b^n \mid n \geq 1\}^\omega$ with $C = \{(z,z) \mid z \in \Nbb\}$.}
	\label{fig:resetpbacounterexample}
\end{figure}

We focus on $L \notin \LPAomega$ and argue by contradiction. 
Suppose, $L \in \LPAomega$, \ie, there are Parikh-recognizable languages $U_1, V_1, \dots, U_n, V_n$ such that $L = U_1 V_1^\omega \cup \dots \cup U_n V_n^\omega$. 
Then there is some $i \leq n$ such that for infinitely many $j \geq 1$ the infinite word $\alpha_j = aba^2b^2 \dots a^jb^j \cdot a^\omega \in U_iV_i^\omega$. 
Then $V_i$ must contain a word of the form $v = a^k$, $k > 0$. 
Additionally, there cannot be a word in $V_i$ with infix~$b$. 
To see this assume for sake of contradiction that there is a word $w \in V_i$ with $\ell = |w|_b > 0$. Let  $\beta = (v^{\ell+1} w)^\omega$. 
Observe that $\beta$ has an infix that consists of at least $\ell+1$ many $a$, followed by at most $\ell$, but at least one $b$, hence, no word of the form $u\beta$ with $u \in U_i$ is in $L$. This is a contradiction, thus $V_i \subseteq \{a\}^+$.

Since $U_i \in \LPA$, there is a PA $\Amc_i$ with $L(\Amc_i) = U_i$. 
Let $m$ be the number of states in~$\Amc_i$ and $w' = aba^2b^2 \dots a^{m^4+1} b^{m^4+1}$. 
Then $w'$ is a prefix of a word accepted by~$\Amc_i$. Now consider the infixes $a^\ell b^\ell$ and the pairs of states $q_1,q_2$, where we start reading $a^\ell$ and end reading $a^\ell$, and $q_3,q_4$ where we start to read $b^\ell$ and end to read $b^\ell$, respectively. 
There are $m^2$ choices for the first pair and $m^2$ choices for the second pair, hence $m^4$ possibilities in total. 
Hence, as we have more than $m^4$ such infixes, there must be two with the same associated states $q_1,q_2,q_3,q_4$. 
Then we can swap these two infixes and get a word of the form $ab \dots a^rb^s \dots a^s b^r \dots a^{m^4+1} b^{m^4+1}$ that is prefix of some word in $L(\Amc_i) = U_i$. 
But no word in $L$ has such a prefix, a contradiction. Thus, $U_1 V_1^\omega \cup \dots \cup U_nV_n^\omega \neq L$.
\end{proof}

%% file: decision.tex
\section{Decision problems}
\label{sec:decision}
In this section, we show that the (un)decidability results of the common decision problems for PA on finite words in \cite{klaedtkeruess, emptynp} can be transferred to PBA. These problems include the following: 
\begin{enumerate}
\item \emph{Emptiness.} Does $\Amc$ recognize the empty language?
\item \emph{Universality.} Does $\Amc$ accept every word?
\item \emph{Equivalence.} Do $\Amc_1$ and $\Amc_2$ recognize the same language?
\item \emph{Inclusion.} Does $\Amc_2$ recognize every word that $\Amc_1$ recognizes?
\end{enumerate}

In the following we assume that all semi-linear sets are given as a collection of lists of base vector $b_0$ and period vectors $b_1,\ldots, b_k$ (one list for each linear set in the finite union). All numbers are encoded in binary. 


\begin{theorem}
Emptiness for SPBA, WPBA, and PPBA (and their equivalent models) are $\coNP$-complete. Universality, inclusion and equivalence for SPBA, WPBA, and PPBA (and their equivalent models) are undecidable.
\end{theorem}

Since every SPBA can be converted efficiently into an equivalent WPBA and vice versa by \Cref{lem:SPBAtoWPBA} and \Cref{lem:WPBAtoSPBA}, the decidability results for SPBA hold for WPBA, too.
Furthermore, we show how to convert any PPBA into an equivalent WPBA (and hence SPBA) in polynomial time (note that this does not directly follow from our results, as the intermediate step of converting a PPBA into an equivalent CM introduces $\varepsilon$-transitions whose elimination requires super-polynomial time).

Thus, for the $\coNP$-completeness of emptiness it is sufficent to show that emptiness for $\varepsilon$-MSPBA is in $\coNP$ (yielding the $\coNP$-membership for all of our models), and $\coNP$-hardness for (even 1-dimensional) PPBA (yielding hardness for all our models).
Similarly, it is sufficient to show the undecidability results for PPBA.

\begin{lemma}
 Let $\Amc = (Q, \Sigma, q_0, \Delta, F, C)$ be a PPBA of dimension $d$. Then there is an equivalent WPBA that can be constructed in polynomial time.   
\end{lemma}
\begin{proof}

 The idea is as follows. We have a copy of $\Amc$ and additionally, for each linear set~$C_i$ in the semi-linear set $C$ we introduce a copy $\Amc_i$ of $\Amc$. 
 All states of the copy of $\Amc$ are non-resetting and all states of $F$ are resetting in the $\Amc_i$. 
 At any point while a run is still in $\Amc$ and would enter a state $f\in F$, we allow a non-deterministic to some $\Amc_i$ and to reset the counters. For this, we first go to a special resetting copy of $f$ and then transition to $\Amc_i$. 
 We use additional counters to mark in the semi-linear set in which $\Amc_i$ we are moving and modify the semi-linear set so that we do not count the base-vector again when continuing the run in $\Amc_i$.


Formally, assume $C=\bigcup_{i\leq k}C_i$. 
 Let $\Amc' = (Q', \Sigma, (q_0, 0), \Delta', F', C')$ of dimension $d+k+1$,
 where 
 \begin{align*}
 Q' =&\ Q\times \{0,1,\ldots, k\}\cup F\times \{1,\ldots, k\}\times \{R\}, \\
 F' =&\ F\times \{1,\ldots, k\}\times \{R\}\cup F\times\{1,\ldots, k\}, \\
\Delta' =&\ \{((p,0), a, \vbf \cdot \0^{k+1}, (q,0)), ((p,i), a, \vbf \cdot \ebf_i^{k+1}, (q,i)) \mid (p, a, \vbf, q) \in \Delta, 1\leq i\leq k\} \\
\cup&\ \{((q,0), a, \vbf \cdot (\ebf_i^{k+1}+\ebf_{k+1}^{k+1}), (f,i,R)\} \mid (q,a,\vbf, f) \in \Delta,f \in F\}\\
\cup&\ \{((f,i,R), a, \vbf \cdot \ebf_i^{k+1}, (q,i)\} \mid (f,a,\vbf, q) \in \Delta, f \in F\}, \text{ and } \\
C' =&\ \bigcup_{i\leq k} C_i\cdot \{(\ebf_i^{k+1}+\ebf_{k+1}^{k+1})\}\cup \bigcup_{i\leq k} \vec{C}_i \cdot \{z\cdot \ebf_i^{k+1}\mid z\in \Nbb\},
 \end{align*}
 where $\vec{C}_i$ is defined as $C_i$ but without any base vectors.
 We prove that $P_\omega(\Amc) = W_\omega(\Amc')$.

\smallskip
 $\Rightarrow$ We first show $P_\omega(\Amc)\subseteq W_\omega(\Amc')$. Let $\alpha \in P_\omega(\Amc)$ with accepting run $r = r_1r_2r_3 \dots$ where $r_i = (p_{i-1}, \alpha_i, \vbf_i, p_i)$, and denote by $f_1< f_2< \dots$ the positions of accepting hits of $r$. By the infinite pigeonhole principle there is some $\ell\leq k$ such that infinitely many of these accepting hits are in the linear set $C_\ell$. Denote by $k_1<k_2<\ldots$ the subsequence of $f_1,f_2,\ldots$ of all accepting hits in $C_\ell$. 

 This means $p_{k_j} \in F$ and $\rho(r_1 \dots r_{k_j}) \in C_\ell$ for all $j \geq 1$. For $i < j$ let $r_{i,j} = r_{k_i} + 1 \dots r_{k_j}$.
 Now observe that $\rho(r_1 \dots r_{k_\ell}) = \rho(r_1 \dots r_{k_1}) + \rho(r_{1,2}) + \rho(r_{2,3}) + \dots +\rho(r_{\ell-1, \ell})$ for all $\ell \geq 1$. In particular, we have $\rho(r_1 \dots r_{k_1}) \in C_\ell$ and $r_{i, i+1} \in \vec{C}_\ell$ for all $i < \ell$.
 
 Let $\hat{r} = \hat{r}_1 \dots \hat{r}_{k_1 - 1} ((p_{k_1 - 1},0), \alpha_{k_1}, \vbf_{k_1} \cdot (\ebf_\ell^{k+1}+\ebf_{k+1}^{k+1}), (f,\ell, R))$, where $\hat{r}_i = ((p_{i-1},0)$, $\alpha_i, \vbf_i \cdot \ebf_0^{k+1}, (p_i,0))$ for all $i \leq k_1-1$, 
 and let $r' = ((f,\ell, R), \alpha_{k_1 + 1}, \vbf_{k_1+1}\cdot \ebf_\ell^{k+1}, (p'_{k_1 + 1},\ell))$ $r'_{k_i + 2} r'_{k_i + 3} \dots$ where $r'_i = ((p'_{i-1},\ell), \alpha_i, \vbf \cdot \ebf_\ell^{k+1}, (p'_i,\ell))$ for all $i \geq k_1+1$.
 Then $\hat{r} r'$ is a run of $\Amc'$ on $\alpha$. Furthermore, this run is accepting, as we can translate the positions of accepting hits of $r$ one-to-one to reset positions of $\hat{r} r'$, and by the choice of $\Delta'$, $C'$ and the observations above. Hence $\alpha \in W_\omega(\Amc')$.

\smallskip
 $\Leftarrow$ To see that $W_\omega(\Amc')\subseteq P_\omega(\Amc)$, let $\alpha \in W_\omega(\Amc')$ with accepting run $r' = r'_1r'_2r'_3 \dots$, 
 where $r'_i = (p_{i-1}', \alpha_i, \vbf_i \cdot u,  p_i')$ where ${p}_i' \in Q'$ and $u \in \{\0^{k+1}, \ebf_\ell^{k+1}, (\ebf_\ell^{k+1}+\ebf_{k+1}^{k+1})\mid$ \mbox{$1\leq \ell \leq k\}$}, and denote by $0 = k_0 < k_1 <k_2 \dots$ the reset positions of $r'$. 
 As no state in the first copy of $\Amc$ is resetting and by the further construction of $\Amc'$, we have $p_i' = (p_i,0)$ for all $i<k_1$, $p_{k_1}'= (p_{k_1},\ell,R)$, and $p_i'=(p_i,\ell)$ for all $i>k_1$ for some $\ell\leq k$. 
 As the transition from $p_{k_1-1}'$ to $p_{k_1}'$ is the only one where the $k+1$rst of the new counters is set to $1$, and such a transition appears only once in $r'$ (there is no way back from the copy $\Amc_\ell$), we have $\rho(r'_1 \dots r_{k_1}') \in C_\ell \cdot \{(\ebf_\ell^{k+1}+\ebf_{k+1}^{k+1})\}$, and for all $i > 1$ we have $\rho(r'_{k_{i - 1} + 1} \dots r'_{k_i}) \in \vec{C}_\ell \cdot \{z\ebf_\ell^{k+1}\mid z\in \Nbb\}$.
 For every $i\geq 1$ let $r_i = (p_{i-1}, \alpha_i, \vbf_i, p_i)$.
Then $r = r_1 r_2 r_3 \dots$ is of $\Amc$ on $\alpha$. Furthermore $r$ is accepting, as the reset positions of $r'$ translate to accepting hits in $r$ one-to-one. 
To see this, observe that for all vectors $\ubf \in C_\ell$ and $\vbf \in \vec{C}_\ell$ we have $\ubf + \vbf \in C_\ell$. As $\rho(r_1 \dots r_{k_1}) = \vbf_1 + \dots + \vbf_{k_1} \in C$ and $\rho(r_{k_{i-1}+ 1} \dots r_{k_i}) = \vbf_{k_{i-1}+1} + \dots + \vbf_{k_i} \in \vec{C}$ for all $i > 1$, the run $r$ is indeed accepting, hence $\alpha \in P_\omega(\Amc)$.

The automaton $\Amc'$ has at most $(k+2)|Q|$ states and can obviously be constructed in polynomial time.
\end{proof}

We now show that non-emptiness for PPBA is $\NP$-hard. A similar proof was sketched in \cite[Proposition III.2]{emptynp}.

\begin{lemma}
Non-emptiness for PPBA is $\NP$-hard.
\end{lemma}
\begin{proof}
We (polynomial-time many-one) reduce from \emph{subset sum}, which is known to be $\NP$-complete \cite{intract}. 
An instance of this problem consists of a finite set $M \subseteq \Nbb$ and a threshold $\ell$, and asks if there is a subset $S \subseteq M$ such that $\sum_{s \in S}s = \ell$. 
Let $\Imc$ be an instance of subset sum with $M = \{m_1, \dots, m_k\}$ and threshold $\ell$. 
Intuitively, we encode a solution of $\Imc$ into a finite word of length $k$ over $\Sigma = \{x_1, \dots, x_k\}$, and pad it using an infinite sequence of dummy symbols $x_D$. 
Afterwards, we construct a PPBA $\Amc_\Imc$ of dimension 1 that accepts valid solutions of $\Imc$ only. 

We choose $\Amc_\Imc = (\{q_i \mid 1 \leq i\leq m\}\, \cup$ $\{q_0, q_f\}, \Sigma \cup \{x_D\}, q_0, \Delta_\Imc, \{q_f\}, \{\ell\})$ where
\begin{align*}
\Delta_\Imc =&\ \{(q_{i-1}, x_i, m_i, q_i), (q_{i-1}, x_i, 0, q_i)  \mid 1 \leq i \leq m\} \\
 \cup&\ \{(q_m, x_D, 0, q_f), (q_f, x_D, 0, q_f)\}.
\end{align*}

We claim that $\Imc$ has a solution if and only if $P_\omega(\Amc_\Imc) \neq \varnothing$.

\smallskip
$\Rightarrow$ Let $S \subseteq M$ be a solution for $\Imc$. It is easy to see that $\Amc_\Imc$ accepts the infinite word $x_1 \dots x_m x_D^\omega$, as the run $r = r_1 \dots r_m (q_m, x_D, 0, q_f) (q_f, x_D, 0, q_f)^\omega$, where $r_i = (q_{i-1}, x_i, m_i, q_i)$ if $m_i \in S$, and $r_i = (q_{i-1}, x_i, 0, q_i)$ otherwise, is accepting.
Thus, $P_\omega(\Amc_\Imc) \neq \varnothing$.

\smallskip
$\Leftarrow$ Let $\alpha$ an infinite word accepted by $\Amc_\Imc$. By the choice of $\Amc_\Imc$ we have $\alpha = x_1\dots x_m x_D^\omega$ with accepting run $r = r_1 \dots r_m (q_m, x_D, 0, q_f) (q_f, x_D, 0, q_f)^\omega$. 
We construct a solution $S$ from $r$, as follows.
For every $1 \leq i \leq m$ we add $m_i$ to $S$ if and only if $r_i = (q_{i-1}, x_i, m_i, q_i)$. Similar to above, the set $S$ is a solution for $\Imc$.

\smallskip
The size of $\Amc_\Imc$ is linear in the size of $\Imc$, hence, this transformation is a polynomial-time reduction.
\end{proof}

In the next step, we show that non-emptiness is in $\NP$ for $\varepsilon$-MSPBA. As proved by Fernau and Stiebe, all CM that accept at least one infinite word, also accept at least one \emph{ultimately periodic} infinite word, that is, a word of the form $uv^\omega$ with $u \in \Sigma^*$ and $v \in \Sigma^+$. 
We show that this is also true for $\varepsilon$-MSPBA (hence for all of our models).
\begin{lemma}
\label{lem:periodic}
    Let $\Amc$ be an $\varepsilon$-MSPBA with alphabet $\Sigma$. If $S_\omega(\Amc) \neq \varnothing$, there is a word of the form $uv^\omega \in S_\omega(\Amc)$ where $u\in \Sigma^*$ and $v \in \Sigma^+$.
\end{lemma}
\begin{proof}
Assume $S_\omega(\Amc) \neq \varnothing$. Then there is an infinite word $\alpha \in S_\omega(\Amc)$ with accepting run $r = r_1 r_2 r_3 \dots$, where $r_i = (p_{i-1}, \gamma_i, \vbf_i, p_i)$. Let $k_1 < k_2 < \dots$ be the positions of all resetting states in $r$. Let $k_i < k_j$ be two such positions such that $p_{k_i} = p_{k_j}$ and $\Amc$ has read at least one symbol from $\alpha$ after leaving $p_{k_i}$ and entering $p_{k_j}$.
Let $u = \gamma_1 \dots \gamma_{k_i}$ be the prefix of $\alpha$ read upon visiting $p_{k_i}$ and $v = \gamma_{k_i + 1} \dots \gamma_{k_j}$ the infix read between $p_{k_i}$ and $p_{k_j}$. Note that $v \neq \varepsilon$ by the choice of $k_j$.
Then $\Amc$ also accepts $uv^\omega$, as $r_1 \dots r_{k_i} (r_{k_i + 1} \dots r_{k_j})^\omega$ is an accepting run of $\Amc$ on $uv^\omega$ by definition.
\end{proof}

This helps giving an $\NP$-algorithm that solves non-emptiness for $\varepsilon$-MSPBA

\begin{lemma}
Non-emptiness for $\varepsilon$-MSPBA is in $\NP$.
\end{lemma}
\begin{proof}
By \Cref{lem:MSPBAtoSPBA} we can convert any $\varepsilon$-MSPBA into an equivalent $\varepsilon$-SPBA (it is easy to see that the construction of the lemma works in polynomial time). Hence, we show how to solve non-emptiness for $\varepsilon$-SBPA.
Let $\Amc = (Q, \Sigma, q_0, \Delta, \Emc, F, C)$ be an $\varepsilon$-SPBA. By Lemma~\ref{lem:periodic}, it suffices to check whether $\Amc$ accepts 
an infinite word $u v^\omega$ with $u \in \Sigma^*$ and $v \in \Sigma^+$. If such a word exists, we may assume that there is an accepting run $r_u r_v^\omega$ of $\Amc$ on $uv$ where neither $r_u$ nor $r_v$ visit the same accepting state twice.
For any $p,q \in Q$ we
define $\Amc_{p \Rightarrow q} = (Q \cup \{q_0'\}, \Sigma, q_0', \Delta', \Emc', \{q\}, C)$,
where $\Delta' = \{(q_1, a, \vbf, q_2) \mid (q_1, a, \vbf, q_2) \in \Delta, q_1 \notin F\}
\cup \{(q_0', a, \vbf, q_2) \mid (p, a, \vbf, q_2) \in \Delta\}$ and, analogously,
\mbox{$\Emc' = \{(q_1, \varepsilon, \vbf, q_2) \mid (q_1, \varepsilon, \vbf, q_2) \in \Emc, q_1 \notin F\}
\cup \{(q_0', \varepsilon, \vbf, q_2) \mid (p, \varepsilon, \vbf, q_2) \in \Delta\}$}.

\smallskip
The following NP algorithm solves non-emptiness:
\begin{enumerate}
\item Guess a sequence $f_1, \ldots, f_k$ of accepting states with $k \leq 2|F|$ such that $f_i = f_k$ for some $i \leq k$.
\item Verify that $L(\Amc_{q_0 \Rightarrow f_1}) \neq \varnothing$ and $L(\Amc_{f_j \Rightarrow f_{j+1}}) \neq \varnothing$ for all $1 \leq j < k$ (interpreted as PA over finite words).\enlargethispage{4mm}
\item Verify that $L(\Amc_{f_i \Rightarrow f_{i+1}}) \cdot \ldots \cdot L(\Amc_{f_{k-1} \Rightarrow f_k}) \not\subseteq \{\varepsilon\}$.
\end{enumerate}

The second step can be done by adding a fresh symbol (say $e$) to the automata and replacing every $\varepsilon$-transition with an $e$-transition (observe that this does construction does not change the emptiness behavior, and is, in contrast to the construction of \Cref{lem:KR-finite}, computable in polynomial time).
Afterwards we use the NP-algorithm for non-emptiness of PA \cite{emptynp}.

The third step essentially states that not all $L(\Amc_{f_j \Rightarrow f_{j+1}})$ for $j \geq i$ may only accept the empty word, as we require $v \neq \varepsilon$. To check this property, we can construct a PA\footnote{This is possible in polynomial time by a construction very similar to the one of \Cref{lem:concatenation}. However, to the best of our knowledge there is no explicit construction for concatenation in the literature for PA on finite words.} recognizing $L(\Amc_{f_i \Rightarrow f_{i+1}}) \cdot \ldots \cdot L(\Amc_{f_{k-1} \Rightarrow f_k})$, and again replace every $\varepsilon$-transition with an $e$-transition. Finally, we build the product automaton with the PA (NFA) that recognizes the language $\{w \in (\Sigma \cup \{e\})^* \mid w \text{ contains at least 1 symbol from } \Sigma\}$ and test non-emptiness for the resulting PA.
\end{proof}

Now we show that universality for PPBA (and thus for all our models) is undecidable.

\begin{lemma}
\label{lem:universality}
Universality for PPBA is undecidable.
\end{lemma}
\begin{proof}
We reduce from the universality problem for PA, which is known to be undecidable~\cite{klaedtkeruess}. Let $\Amc = (Q, \Sigma, q_0, \Delta, F, C)$ be an arbitrary PA. We construct a PPBA $\Amc'$ that is universal if and only if $\Amc$ is universal.
The idea is the following. The PPBA $\Amc'$ accepts all words that are accepted by $\Amc$, followed by an infinite sequence of dummy symbols~$x_D$. Furthermore, $\Amc'$ accepts all words that are not a member of $\Sigma^* \cdot \{x_D\}^\omega$, \ie, all ``malformed" words. In the first step, we construct a PPBA $\Amc_1$, where 
\[\Amc_1 = (Q \cup \{q_f\}, \Sigma \cup \{x_D\}, q_0, \Delta_1, \{q_f\}, C)\] 

\vspace{-3mm}
with 
\vspace{-3mm}

\[\Delta_1 = \Delta \cup \{(q, x_D, \0, q_f) \mid q \in F\} \cup \{(q_f, x_D, \0, q_f)\}.\]

Observe that $P_\omega(\Amc_1) = L(\Amc) \cdot \{x_D\}^\omega$. 
Let $\Amc_2$ with $P_\omega(\Amc_2) = \Sigma^\omega$ and let $\Amc_3$ be a PPBA that accepts all infinite words that contain a dummy symbol followed by a non-dummy symbol, \ie, $P_\omega(\Amc_3) = (\Sigma \cup \{x_D\})^* \cdot \{x_D\} \cdot \Sigma \cdot \{\Sigma \cup \{x_D\}\}^\omega$ (note that these languages are even $\omega$-regular). Let $L = L_\omega(\Amc_1) \cup L_\omega(\Amc_2) \cup L_\omega(\Amc_3)$. According to \Cref{cor:closure}, we can construct a PPBA that recognizes $L$. Let $\Amc'$ be this PPBA. Then $L(\Amc) = \Sigma^*$ if and only if $P_\omega(\Amc') = (\Sigma \cup \{x_D\})^\omega$.

\smallskip
$\Rightarrow$  Let $\Amc$ be universal, \ie, $\Amc$ accepts all words $w \in \Sigma^*$, that is, according to the construction of $\Amc_1$, all infinite words of the form $w \cdot x_D^\omega$ will be accepted by $\Amc'$. Since $\Amc_2$ and $\Amc_3$ accept all infinite words in $(\Sigma \cup \{x_D\})^\omega \setminus (\Sigma^* \cdot \{x_D\}^\omega)$, the PPBA $\Amc'$ accepts all infinite words over $\Sigma \cup \{x_D\}$, \ie, $\Amc'$ is universal.

$\Leftarrow$ Let $\Amc'$ be universal. We argue by contradiction. Suppose $\Amc$ is not universal, \ie, there is a word $w \in \Sigma^*$ that is rejected by $\Amc$. Then $\alpha = w \cdot x_D^\omega$ will be rejected by $\Amc_1$. Moreover, $\Amc_2$ and $\Amc_3$ do not accept $\alpha$ either. Hence, $\alpha \not \in P_\omega(\Amc')$, a contradiction.
\end{proof}
\begin{remark}
 This proof works similarly if we interpret 
 $\Amc'$ as an SPBA or WPBA. The only difference is that we consider the normalized automaton $\Amc_N$ of $\Amc$ and check the membership of $\varepsilon$ separately, which is still computable as the membership problem for PA is decidable (in NP) \cite{klaedtkeruess, emptynp}.
\end{remark}

Knowing that universality is undecidable for PPBA, we can derive the undecidability for equivalence and inclusion easily as follows: Let $\Sigma$ be an arbitrary alphabet and $\Amc_\Sigma$ be a PPBA with $P_\omega(\Amc_\Sigma) = \Sigma^\omega$. Let $\Amc$ be an arbitrary PPBA. We observe that $\Amc$ is universal if and only if $P_\omega(\Amc) = P_\omega(\Amc_\Sigma)$ if and only if $P_\omega(\Amc_\Sigma) \subseteq P_\omega(\Amc, C)$.
Again, these observations hold for SPBA, WPBA, and all equivalent models.

%% file: conclusion.tex
\section{Conclusion and Outlook}\label{sec:conclusion}
We have introduced Parikh-Büchi automata with different acceptance conditions.
We have shown that PBA with strong and weak reset acceptance condition describe the same class of $\omega$-languages, which is a strict superclass of the class PPBA-recognizable $\omega$-languages. 
The latter class is equivalent to the class of blind $k$-counter machines introduced by Fernau and Stiebe. As a side-product we show that all our models admit $\varepsilon$-elimination and are equivalent to their multi-counterparts. Finally, we have studied common decision problems.

An interesting open question remains the search for an (intuitive) model that captures exactly~$\LPAomega$. Vice versa, characterizations of $\LPrefix$ and $\LReset$ in the spirit of Büchi's theorem are yet to be found. Furthermore, it is interesting to study deterministic PBA, as deterministic Büchi automata are already weaker than (non-deterministic) Büchi automata. Finally, Klaedtke and Ruess proposed to study a variant of PBA where the semi-linear sets are augmented with a symbol $\infty$ for infinity, and the extended Parikh image of an infinite word is computed accordingly using transfinite induction.

%% file: main-infinite.bbl
\begin{thebibliography}{10}

\bibitem{buechi}
J.~Richard Büchi.
\newblock Weak second-order arithmetic and finite automata.
\newblock {\em Mathematical Logic Quarterly}, 6(1‐6):66--92, 1960.

\bibitem{DBLP:journals/ita/CadilhacFM12}
Micha{\"{e}}l Cadilhac, Alain Finkel, and Pierre McKenzie.
\newblock Affine parikh automata.
\newblock {\em {RAIRO} Theor. Informatics Appl.}, 46(4):511--545, 2012.

\bibitem{DBLP:journals/ijfcs/CadilhacFM12}
Micha{\"{e}}l Cadilhac, Alain Finkel, and Pierre McKenzie.
\newblock Bounded parikh automata.
\newblock {\em Int. J. Found. Comput. Sci.}, 23(8):1691--1710, 2012.

\bibitem{DBLP:conf/fossacs/DartoisFT19}
Luc Dartois, Emmanuel Filiot, and Jean{-}Marc Talbot.
\newblock Two-way parikh automata with a visibly pushdown stack.
\newblock In Mikolaj Bojanczyk and Alex Simpson, editors, {\em Foundations of
  Software Science and Computation Structures - 22nd International Conference,
  {FOSSACS} 2019, Held as Part of the European Joint Conferences on Theory and
  Practice of Software, {ETAPS} 2019, Prague, Czech Republic, April 6-11, 2019,
  Proceedings}, volume 11425 of {\em Lecture Notes in Computer Science}, pages
  189--206. Springer, 2019.

\bibitem{blindcounter}
Henning Fernau and Ralf Stiebe.
\newblock Blind counter automata on omega-words.
\newblock {\em Fundam. Inform.}, 83:51--64, 2008.

\bibitem{emptynp}
Diego Figueira and Leonid Libkin.
\newblock Path logics for querying graphs: Combining expressiveness and
  efficiency.
\newblock In {\em Proceedings of the 2015 30th Annual ACM/IEEE Symposium on
  Logic in Computer Science (LICS)}, LICS ’15, pages 329--340, Kyoto, Japan,
  2015. IEEE.

\bibitem{FiliotGM19}
Emmanuel Filiot, Shibashis Guha, and Nicolas Mazzocchi.
\newblock Two-way parikh automata.
\newblock In Arkadev Chattopadhyay and Paul Gastin, editors, {\em 39th {IARCS}
  Annual Conference on Foundations of Software Technology and Theoretical
  Computer Science, {FSTTCS} 2019, December 11-13, 2019, Bombay, India}, volume
  150 of {\em LIPIcs}, pages 40:1--40:14. Schloss Dagstuhl - Leibniz-Zentrum
  f{\"{u}}r Informatik, 2019.

\bibitem{intract}
Michael~R. Garey. and David~S. Johnsons.
\newblock {\em Computers and Intractability: A Guide to the Theory of
  NP-Completeness}.
\newblock W. H. Freeman, New York, USA, 1st edition, 1979.

\bibitem{klaedtkeruess}
Felix Klaedtke and Harald Rue{\ss}.
\newblock Monadic second-order logics with cardinalities.
\newblock In Jos C.~M. Baeten, Jan~Karel Lenstra, Joachim Parrow, and
  Gerhard~J. Woeginger, editors, {\em Automata, Languages and Programming},
  pages 681--696, Berlin, Heidelberg, 2003. Springer.

\bibitem{parikh1966context}
Rohit~J Parikh.
\newblock On context-free languages.
\newblock {\em Journal of the ACM (JACM)}, 13(4):570--581, 1966.

\bibitem{thomas2002automata}
Wolfgang Thomas et~al.
\newblock {\em Automata, logics, and infinite games: a guide to current
  research}, volume 2500.
\newblock Springer Science \& Business Media, 2002.

\end{thebibliography}
